%% file: Remote_Classification_Arxiv_Submit.tex
\documentclass[draftcls, onecolumn, 12pt]{IEEEtran}
\IEEEoverridecommandlockouts
\usepackage{cite}
\usepackage{amsmath,amssymb,amsfonts}
\usepackage{algorithmic}
\usepackage{graphicx}
\usepackage{textcomp}
\usepackage{xcolor}
 \usepackage{amsmath}
 \usepackage{amssymb}
 \usepackage{subfigure}
 \usepackage{graphicx}
 \usepackage{graphics}
 \usepackage{color}
 \usepackage{psfrag}
 \usepackage{cite}
 \usepackage{balance}
 \usepackage{algorithm}
 \usepackage{accents}
 \usepackage{amsthm}
 \usepackage{bm}
 \usepackage{url}
 \usepackage{algorithmic}
 \usepackage[english]{babel}
 \usepackage{multirow}
 \usepackage{enumerate}
 \usepackage{cases}
 \usepackage{stfloats}
 \usepackage{dsfont}
 \usepackage{soul}
 \usepackage{amsfonts}
 \usepackage{fancyhdr}
 \usepackage{hhline}
 \usepackage{array}
 \usepackage{booktabs}
 \usepackage{framed}
 \usepackage{tensor}
 \usepackage{setspace}

\include{header}

\def\BibTeX{{\rm B\kern-.05em{\sc i\kern-.025em b}\kern-.08em
    T\kern-.1667em\lower.7ex\hbox{E}\kern-.125emX}}

\begin{document}

\title{Capacity of Remote Classification Over Wireless Channels}

\author{\IEEEauthorblockN{Qiao Lan},
\and
\IEEEauthorblockN{Yuqing Du},
\and
\IEEEauthorblockN{Petar Popovski},
\and and
\IEEEauthorblockN{Kaibin Huang} 
\thanks{\setlength{\baselineskip}{13pt} \noindent Q. Lan, Y. Du, and K. Huang are with the Dept. of Electrical and Electronic Engineering at The  University of  Hong Kong, Hong Kong (Email: qlan@eee.hku.hk, yqdu@eee.hku.hk, huangkb@eee.hku.hk). Petar Popovski is with the Dept. of Electronic Systems at the Aalborg University, Aalborg, Denmark (Email: petarp@es.aau.dk). Corresponding author: K. Huang. } }

\maketitle

\vspace{-5.5mm}
\begin{abstract}
Wireless connectivity creates a computing paradigm that merges communication and inference. A basic operation in this paradigm is the one where a device offloads classification tasks, such as object recognition, to the edge serves. We term this 
\emph{remote classification}, with a potential to enable  many intelligent applications ranging from autonomous driving to augmented reality. Remote classification is challenged by the finite and variable data rate of the wireless channel, which affects the capability to transfer high-dimensional features and thus limits the classification resolution.
We introduce a set of metrics under the name of \emph{classification capacity} that are defined as the maximum number of classes that can be discerned over a given communication channel while meeting a target probability for classification error. We treat both the cases of a channel where the instantaneous rate is known  and unknown. The objective is to choose a subset of classes from a class library that offers satisfactory performance over a given channel. We treat two different cases of subset selection. \emph{First,} a device can select the subset by pruning the class library until arriving at a subset that meets the targeted error probability while maximizing the classification capacity. Adopting a subspace data model, we prove the equivalence of classification capacity maximization to the problem of packing on the Grassmann manifold. The results show   that  the classification capacity grows \emph{exponentially} with the instantaneous communication rate, and \emph{super-exponentially} with the dimensions of each data cluster. This also holds for 
ergodic and outage capacities with fading if the instantaneous rate is replaced with an average rate and a fixed rate, respectively. In the \emph{second} case, a device has a unique preference of class subset for every communication rate, which is modeled as an instance of uniformly sampling the library. Without class selection, the classification capacity and its ergodic and outage counterparts are proved to scale \emph{linearly} with their corresponding communication rates  instead of the exponential growth in the last case. 
\end{abstract}
%============================================================

\section{Introduction}

There is an emerging trend of deploying various \emph{Artificial Intelligent} (AI) algorithms at the edge, away from the central cloud, to provide a context-aware and low-latency platform for supporting a wide range of applications such as Internet search (e.g., Google Lens), digital payment (e.g., Alipay's Smile to Pay), and inter-connected vehicles in 5G \cite{Niyato2020CommSurvey}. The ubiquitous wireless connectivity  results in a new paradigm merging communication and inference, called \emph{edge inference}, referring to the broad set of techniques for deploying trained  AI models  at edge servers to remotely execute inference  tasks posed by mobile users, such as object recognition or speech interpretation. 

A large class of edge inference services can be reduced to the model in which a mobile user wirelessly uploads a multimedia data sample (photo, video or speech clip) over a wireless link, the edge server recognizes an object embedded in the sample and feeds back the object label. We term this operation \emph{remote classification} and it is the main theme of this work. A large-scale remote classifier in the edge/central cloud is capable of rendering many object classes, around 700 for Google Cloud and 200 for Tencent Cloud. Maximizing the correctness of classification requires a user to upload high-dimensional features (or large-size raw data). However, this is challenged, on the one hand, by the variability of the wireless link due to fading and interference and, on the other hand, by the stringent latency requirements in real time and/or high-mobility applications. To address this issue, an existing remote classification service achieves the required versatility by deploying  a system of classifiers with diversified capacities, which are switched according to the application requirements, input data quality, or dimensionality of input feature vectors. Motivated by this, we study the capacity of remote classification as a function of the communication rate offered by the wireless link. 

\subsection{Classification, Channel Coding, and Source Coding}

%--------Revision---------------July 12

The essence of the remote classification problem can be better understood by relating it to two classic problems in information theory: source coding and channel coding. In source coding (or compression), a transmitter represents source information using codewords that can be sent over a limited-rate channel and enable the receiver to accurately reconstruct the information~\cite{Berger1971Book}. In channel coding, the transmitter selects a set of codewords to which the messages are mapped 
and the receiver should be capable of differentiating the codewords even in the presence of channel noise, thereby decoding transmitted messages~\cite{Cover1991Book}. Source coding can be seen as a process of remote estimation, while channel coding as a process of remote classification in which codewords are subject to design. Remote classification can be related to coding  based on the following interpretation. We can view class as ``codewords'' chosen by the nature  and the objects as noisy instances of the classes~\cite{Bishop2006Book}. Then the transmitter sends a description (features) of a noisy instance over a limited-rate channel, such that the receiver is able to ``decode" (infer) the ``codeword"  (the covert class or the label of the instance). Therefore, the name of ``remote classification'' as used in this paper, refers to a particular remote classification process in which the classes (``codewords'') are not subject to design.

Despite the similarity, there exist several fundamental differences between remote classification and source/channel coding. First, the ``codewords'' (classes) in the former are chosen by the nature and not subject to design as in the latter. As a result, a typical {multimedia} classifier cannot be derived theoretically. Instead, it is usually computed using a supervised machine learning technique, which includes choosing a suitable model [e.g., \emph{support vector machine} (SVM) or \emph{convolutional neural network} (CNN)] and training the model using a large labeled dataset~\cite{Bishop2006Book}. Second, in source/channel coding, it is the transmitter that has the ground-truth information while the receiver gets an imperfect version of this information. In contrast, in remote classification, the receiver is the one responsible for inferring the ground-truth information (in the form of a label); the transmitter does not have the information and acquires it via a feedback channel. Finally, the general problem of multimedia classification can have different mathematical characteristics from those of coding such as data spaces (e.g., a feature space versus a Galois field).

Despite the differences, relating remote classification to source/channel coding creates the possibility of exploiting analytical tools from the rich literature on the latter to study the former. An early work in this direction is~\cite{Nokleby2015TIT} where the rate of a stand-alone classifier is found to be mathematically equivalent to the capacity of a MIMO channel with space-time modulation. In this work, we adopt a similar approach to investigate the performance of a different system of remote classification featuring a pair of separated classifier and data source that are connected using a wireless channel.

\subsection{Edge Computing and Inference}

Remote classification and edge inference at large are services supported on the edge computing architecture~\cite{Niyato2020CommSurvey}. The current work shares the same spirit as that on computation offloading, a main theme of edge computing research, where mobile devices use unreliable wireless links to offload computation to edge servers. In the  current work, this computation is exemplified by classification. Edge computing augments the capabilities of mobile devices while preserving their energy efficiency~\cite{You2017CommSurvey}. To reduce the devices' energy consumption, a key approach for energy efficient computation offloading is to jointly optimize radio resource allocation to multiple users and their offloaded computation loads  \cite{You2017TWC,Quek2017TCOM,Niyato2012TWC}. Stochastic optimization tools, such as Lyapunov optimization, are applied to adapting offloading decisions \cite{Niyato2012TWC} and servers' CPU frequencies \cite{Zhang2017TWC} to the dynamics in computation tasks and channels in order to reduce both  latency and power consumption. More complex techniques for accelerating offloaded computation include replicated computation at multiple servers \cite{Tao2020TWC}, adding the new dimension of caching to the joint communication-and-computation control \cite{Niyato2020TMC}, and scheduling of computation tasks \cite{Zhang2016ISIT}. Without considering a specific application, the prior work is based on generic computation models, such that the load is measured by the number of bits and the speed by the number of bits computed per second.

Attempts on materializing the vision of edge AI has led to the emergence of edge learning (see, e.g., \cite{Bennis2020Arxiv,Niyato2019Arxiv}, for an overview) and edge inference, which is the theme of this work. Research in edge inference has resulted in several interesting design approaches. Building on the mentioned idea of replicated computation in \cite{Tao2020TWC},  it is proposed in \cite{Shi2020JCIN} that the association between servers (base stations) and devices can be optimized together with beamforming to reduce the total energy consumption of the devices.  Several research groups have developed techniques to implement device-edge cooperative inference, where a learning task is partitioned and executed partially on device and partially offloaded to the servers \cite{Zhou2019Arxiv,Zhang2020Arxiv,Chen2018MECOMM}. To address the issue of limited computation capacity of a device,  a CNN model can be pruned before  partitioning,  and the idea can be implemented using the techniques in  \cite{Zhou2019Arxiv}. There also exist techniques for channel adaptive model partitioning and  coding \cite{Zhang2020Arxiv}. Furthermore, the model partitioning can be adjusted according to the allocated bandwidth and the requirements on latency and inference accuracy, which is  the approach advocated in \cite{Chen2018MECOMM}. In addition, data compression for communication-efficient edge inference has also been investigated. For example, a relevant architecture is proposed in \cite{Katti2018HotNets} where a \emph{deep neural network} (DNN) encoder is deployed at a transmitter to compress raw data and the compressed data is decoded by the server using a DNN encoder before feeding the output into another DNN model for inference. In view of  prior works, they are  focused on technique design and rely  on experiments for performance evaluation. There exist few results on the fundamental limits of edge inference systems under the constraint of  wireless channels connecting severs and devices, which motivates the current work. 

\subsection{Contributions and Organization}

The objective  of this work is to make the first attempt on quantifying the performance of a remote classification system under a \emph{communication channel constraint}, referring to the finite and time-varying rate of a wireless communication channel. To this end, we consider a system in which a mobile device sends a feature vector over a wireless channel to an edge server, which performs 
classification and sends the result to the mobile device. The server supports classification of an arbitrary subset of a class library based on a mainstream architecture of large-scale classification ({see e.g., \cite{Yu2017SIGIR}}). On the one hand, even if the communication rate is sufficiently large for transmitting all features of each sample, classification errors can still occur as an inherent effect of \emph{data noise}, which is caused by the natural factors in sensing (e.g., pose, perspective, lighting, and background). On the other hand, as the rate varies and so does the received number of features per sample, if the classification error probability should be constrained, the maximum number of object classes that are chosen to be discerned by the remote classifier has to be adapted to the rate in a similar way as the maximum constellation order of adaptive modulation. This gives rise to a performance metric called \emph{$\epsilon$-classification capacity} defined as the maximum number of classes that can be discriminated under the channel constraint and for a given target classification error probability\footnote{In the following text, it will be always implicitly  assumed that there is a target classification probability that needs to be met.}. Furthermore, two derivative metrics, called \emph{ergodic} and \emph{outage classification capacities}, are defined to account for the effect of fading, which correspond to adaptive and fixed coding rates, respectively. Using these metrics, the system performance is analyzed for two cases.
\begin{itemize}
\item \textbf{Class-selection case}:  The user has broad interests covering the whole library (e.g., augmented reality). Given a communication rate, the user selects a subset of classes for classification with the aim of maximizing the classification capacity while meeting the target error probability.

\item \textbf{Random-class case}: The user makes a unique choice of class subset for every given communication rate. The subset is modeled as an instance generated by uniform sampling of the library. 
\end{itemize}

The scope of contributions made by this work is described as follows. For tractability, we follow the relevant work in~\cite{Nokleby2015TIT} to adopt  a statistical  data model from the area of linear regression and a matching subspace \emph{maximum-likelihood} (ML)  classifier. Though an alternative classifier model, namely {a neural network}, is also considered in experiments, tractable analysis of its classification capacity remains an open problem. For the current analysis, it is sufficient to use a generic model of wireless channel characterized by a time varying rate. Specific physical-layer techniques such as MIMO, OFDM and NOMA for supporting the rate are not explicitly considered.

The main contributions of the work are summarized as follows. 
\begin{itemize}
\item \textbf{Classification capacity with class selection:} Consider the mentioned class-selection case. For a large library, the problem of maximizing the $\epsilon$-classification capacity by class selection is shown to be equivalent to the mathematical problem of packing on a Grassmann manifold. The relation allows the application of packing results together with error probability analysis of space-time modulation to derive bounds on the maximum capacity. The results reveal  the \emph{exponential growth} of the capacity  with the communication rate and \emph{super-exponential growth} with the dimensions of each data cluster. Based on the results and considering Rayleigh fading, ergodic and outage classification capacities are proved to  follow the same scaling laws as stated above if the communication rate is replaced by its ergodic counterpart or the maximum rate under an outage constraint.

\item \textbf{Classification capacity with random classes:} Consider the other case of  random classes. The  expected classification  error probability is related to the isotropic distribution on a Grassmann manifold. Applying relevant results allow the derivation of a lower bound on the classification capacity, which increases \emph{linearly} with the communication rate. Lower bounds on ergodic and outage classification capacities with Rayleigh fading are also derived and shown  to follow the same scaling law.

\item \textbf{Extension to fast fading:} The preceding results based slow fading are extended to the case with fast fading, resulting in a \emph{random} number of features used for remote classification  of  each  data sample. It is found that fast fading does not change the classification-capacity scaling laws except for adding to  the communication rates the multiplicative factor equal to some packet-success probability. 

\item \textbf{Experiment results:} Experiments based on both the  statistical data model and  a real dataset (MNIST) are conducted to demonstrate the effects of wireless channel on the capacities of remote classification and the classification capacity gains of the class selection case \emph{with respect to} (w.r.t.) the random-class case. 

\end{itemize}
\emph{Organization}: The remainder of the paper is organized as follows. The models and performance metrics are introduced in Section~\ref{sec: system_model}. Section~\ref{sec:class_selection} presents the analysis on classification capacities with class selection while that for the random-class case is investigated in Section~\ref{sec:Pe_random_data}. The derived results are further extended to fast fading channels in Section~\ref{Section:Fast_fading}. Section~\ref{sec:experiments} provides the experimental results, followed by concluding remarks in Section~\ref{Sec:Concluding Remarks}.

\section{Models and Metrics}
\label{sec: system_model}
Consider the remote classification system in Fig.~\ref{fig: system_model}, where an edge device transmits feature vectors, extracted from data samples, to an edge server for classification using a trained model and receives from the server the inferred labels. The specific models and performance metrics are described as follows.

\begin{figure}[t]
\centering
\includegraphics[scale=0.25]{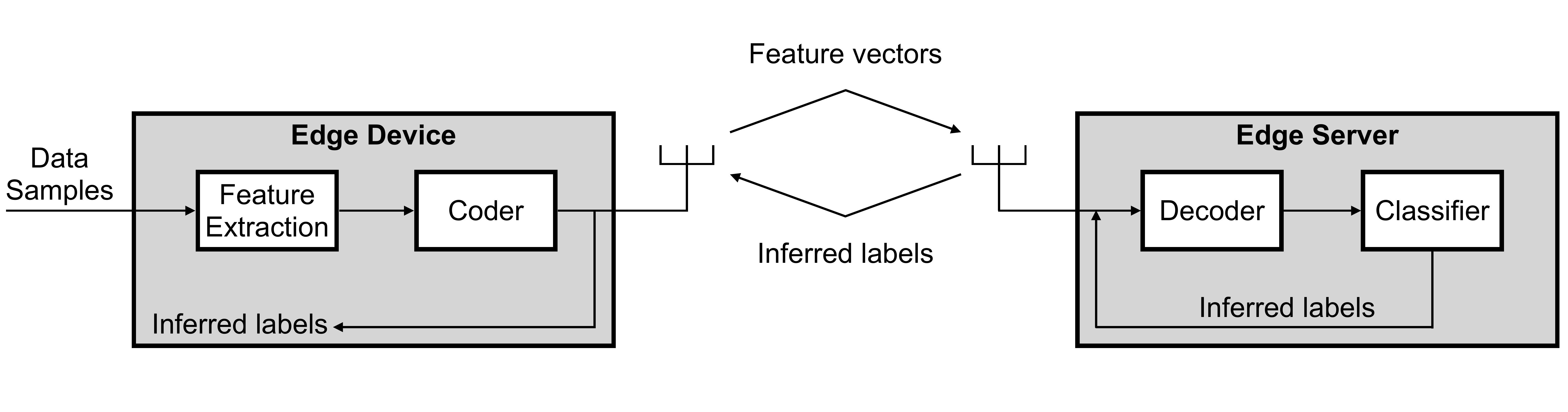}
\caption{Remote classification system.}
\label{fig: system_model}
\end{figure}

\subsection{Classification Model}

As in~\cite{Nokleby2015TIT}, we consider the classic statistical problem of classifying linear subspaces. The statistical data model and ML classifier are described as follows.

\subsubsection{Statistical data model}
Consider a clustered dataset comprising $L$ separable classes, where the $i$-th class centroid is represented by a unitary matrix $\bU_i \in \mathbb R^{N\times K}$ with  $N\geq K$ and  $\bU_i^{T}\bU_i=\bI_K$. An arbitrary data sample, denoted as $\tilde{\bx}$, that  belongs to the $i$-th class is modeled as  \cite{Nokleby2015TIT}
\begin{equation}
\tilde{\bx} = \Phi\bU_i\bs + \tilde{\bw},
\end{equation}
where the unitary matrix $\Phi$ represents the discriminant subspace embedded in the raw-data space, $\bs$ results  from the projection of the data sample into the class subspace $\bU_i$, and $\tilde{\bw}$ accounts for both the error in fitting the dataset distribution to the  subspace model as well as the mentioned data noise. Note that $\tilde{\bw}$ is the inherent cause of classification errors even in the absence of channel constraint. The random vector $\bs  \in \mathbb R^{K}$ is assumed to consist of  \emph{independent and identically distributed} (i.i.d.) $\mathcal N(0,\sigma^2_{\bs})$ elements.  To compress the sample, a feature vector, denoted as  ${\bx}$, is extracted from $\tilde{\bx}$ by projecting it onto the  discriminant subspace: 
\begin{equation}
\label{signal_model}
\bx = \Phi^{ T} \tilde{\bx} = \bU_i\bs + \bw, 
\end{equation}
where ${\bw}=\Phi^T\tilde\bw$ comprises i.i.d. $\mathcal N(0,\sigma^2_{\bw})$ elements and referred hereafter simply as data noise. The subspace $\Phi$ is assumed to be known to the server for calibrating the needed classifier; when $\Phi$ is determined by the sever, the operation is known in the literature as \emph{feature selection}. Based on  \eqref{signal_model}, the data model can be parameterized by the  subspace set $\mathcal U_{L} = \{\bU_{\ell}\}$.

\begin{definition}(Data SNR).
\emph{
The \emph{signal-to-noise ratio} (SNR) of the dataset is defined as the ratio between the variance of each cluster and that of data noise:
\begin{equation}
{\sf Data\ SNR}= \frac{\sigma_{\bs}^2}{\sigma^2_{\bw}} = \sigma_{\bs}^2,
\end{equation}
where we set  $\sigma^2_{\bw} = 1$ without loss of generality.}
\end{definition}

\subsubsection{Maximum-likelihood remote classifier}
Conditioned on $\mathbf{U}_i$, the \emph{probability density function} (PDF) of $\mathbf{x}$ is given as
\begin{align}\label{eqn:detection_ml_v}
P(\mathbf{x}\ \vert \ \mathbf{U}_i )
=\frac{\exp{\left(-\frac{1}{2}\bx^{T}{\left(\sigma_{\bs}^2\bU_i\bU_i^{T}+\bI_N\right)^{-1}\bx}\right)}}{\left(2\pi\right)^{N/2}\det^{1/2}\left(\sigma_{\bs}^2\bU_i\bU_i^{T}+\bI_N\right)} =\frac{\exp{\left(-\frac{1}{2}\bx^{T}\bx+\frac{\sigma_{\bs}^2}{2\left(1+\sigma_{\bs}^2\right)}\bx^{T}\bU_i\bU_i^{T}\bx\right)}}{\left(2\pi\right)^{N/2}\left(1+\sigma_{\bs}^2\right)^{K/2}}.
\end{align}
Given  the knowledge of $\{\bU_i\}_{i=1}^{L}$, the classifier estimates  the class of a reliably received feature vector $\bx$, say $\bU_i$, (or equivalently the label $i$) by maximizing the above PDF: 
\begin{align}\label{eqn:pe_ml_v}
\widehat{i} \triangleq \arg\max_{i\in\{1,2,...,L\}}  p({\mathbf x}|\mathbf{U}_i)= \arg\min_{i\in\{1, 2, ..., L\}} \bx^{T}\bU_i\bU_i^{T}\bx,
\end{align}
which is  the well-known ML classifier. 

\begin{remark}(Geometric Interpretation).
 \emph{
The operation  $\bU_i\bU_i^{T}\bx$ projects the feature vector  $\bx$ onto the subspace, $\text{span}\{\bU_i\}$. It gives the geometric interpretation that  the ML classifier essentially aims at identifying  the subspace forming  the smallest angle  with (or equivalently having the smallest subspace distance to)  the  feature vector $\bx$.
}
\end{remark}

\subsection{Communication Model}
Time is divided into slots, each of which  has the duration of $T$ seconds.  Each feature is quantized into  a sufficiently large number of bits, denoted as $Q$, such that distortion is negligible.  The channel code is designed such that each  quantized feature vector is encoded into a single codeword transmitted using one slot. The  variation of the channel with bandwidth $B$ is assumed to be slow w.r.t. the slot duration such that the channel remains constant within each slot but varies over slots. The extension to the scenario of fast channel variation is presented in Section~\ref{Section:Fast_fading}. Let $R$ denote the communication rate (bit/s) of the channel. Both the cases of channel adaptive and fixed coding rates are considered as discussed in the sequel.  As an example, given the transmit SNR (denoted as $\rho$) and without \emph{channel state information at the transmitter} (CSIT), the rate for a \emph{single-input-single-output} (SISO) channel is $R = B\log_2(1+ \rho|h|^2)$, where $h$ denotes the channel gain. As another example,  the rate for a \emph{multiple-input-multiple-output} (MIMO) channel is $R =  B\log_2\det \left| \mathbf{I}+ \frac{\rho}{N_t}\mathbf{H}\mathbf{H}^H\right|$ where $\mathbf{H}$ denotes the channel matrix and $N_t$ the number of transmit antennas.  The finite communication rate introduces  a constraint on the feature dimension $N =  \beta R$  where $\beta = \frac{T}{Q}$.

\subsection{Performance Metrics}\label{Subsection: model}
To facilitate defining the performance metrics, the notion of (object) \emph{class library} is first formalized. In practical remote-classification, the server supports a large library of classes and can generate an active classifier for a user based on the chosen subset of classes (see e.g., \cite{Yu2017SIGIR}). The class library is represented by $\mathcal{F}=\{\bF_1,\bF_2,...,\bF_M\}$ where each element is a subspace matrix representing an available class. The $L$-class subset chosen by a user is specified by the subspace set $\mathcal{U}_{L}$ with $\mathcal{U}_{L} \subset \mathcal{F}$, which determines  the dataset distribution. 

\subsubsection{Classification error probability}
Labels inferred by the remote classifier can be erroneous as an inherent effect of data noise even though the channel is reliable and even if its rate is sufficiently large to transfer all features. A classification error is declared if the inferred label is different from the ground truth. Conditioned on the data distribution specified by $\mathcal{U}_{L}$ and  the communication rate $R$, the classification error probability, denoted as  $P_{\sf e}$, can be written as
\begin{equation}
\label{eqn: Pe_conditioned_on_R}
P_{\sf e}(R, \mathcal{U}_{L}) \triangleq \frac{1}{L}\sum_{\ell=1}^{L}{\Pr\left( \mathcal{L}(\bx)\neq \ell \ \vert \ y = \ell, \mathcal{U}_{L}, R \right)},
\end{equation}
where $\mathcal{L}$ denotes the classifier  function mapping the input feature vector to the inferred label and    $y$ is the ground-truth label. Note the above definition assumes that the prior probability that the object $\bx$ belongs to one of the $L$ classes is uniform, as in~\cite{Nokleby2015TIT} and $R$ determines the length of $\bx$ as described in the sequel. The future extension to the case with  non-uniform probabilities requires modifying the classifier model by adding prior-dependent weights to the likelihoods of different labels.

\subsubsection{$\epsilon$-classification capacity} Recall that the metric, denoted as $C$, is defined as the maximum number of classes that can be discriminated given an instantaneous communication rate, $R$, such that the classification error probability, $P_{\sf e}$, is no larger than a given threshold $\epsilon \in (0,1)$.  Conditioned on $R$ and $\mathcal U_{L}$, the error probability can be written as the function $P_{\sf{e}}(R,\mathcal U_{L})$. Using the notation, the \emph{$\epsilon$-classification capacity} for the class-selection case can be defined as:
\begin{equation}\label{Def: Classification_Capacity_sel}
C^{\sf sel}(R) = \sup_{\mathcal U_{L}\in \mathcal F,L}\l\{L \ |\ P_{\sf e}(R,\mathcal U_{L}) \leq \epsilon\r\}.
\end{equation}
The counterpart for the random-class case is defined as 
\begin{equation}\label{Def: Classification_Capacity_rnd}
C^{\sf rnd}(R) = \sup_{L}\l\{L \ |\ \mathbb E_{\mathcal U_L}\l[P_{\sf e}(R,\mathcal U_{L})\r] \leq \epsilon\r\},
\end{equation}
where the expectation is over the distribution of the classes, $\mathcal U_L$, given $L$. 

\subsubsection{Ergodic classification capacity}

Consider the case where the device has CSIT and adapts the number of features per sample as well as coding rate to the channel state. Then we can define the \emph{ergodic classification capacity} as:
\begin{align}\label{eqn: ergodic_capacity_definition}
\bar{C} = \left\{ 
 \begin{array}{rcl}
 \mathbb E_R\l[C^{\sf sel}(R)\r], &  & \text{class-selection case};\\
  \mathbb E_R\l[C^{\sf rnd}(R)\r], &  &  \text{random-class case},
\end{array}
\right. 
\end{align}
where the expectations are over the distribution of communication rate $R$,  $C^{\sf sel}(R)$ is defined in~\eqref{Def: Classification_Capacity_sel} and $C^{\sf rnd}(R)$ in~\eqref{Def: Classification_Capacity_rnd}.

\subsubsection{Outage classification capacity}
A different communication model is adopted  where either the CSIT  is unavailable  or some form of channel inversion is used such that the channel cannot be inverted when its gain is below a given threshold. As a result,  the device fixes the number of features per sample and coding rate, resulting in a required communication rate, $r$, for successful decoding of a received feature vector. Then a channel outage event is one that the channel capacity falls below a given threshold  $r$, yielding  the  outage probability defined as 
\begin{equation}\label{Def: Channel_Outage}
P_{\sf{out}}(r) \triangleq \Pr(R\leq r).
\end{equation}
Under an outage constraint, $P_{\sf{out}} \leq \delta$,  and a fixed transmit SNR, there exists  a maximum rate of $r$. Then the outage classification capacity is defined as 
\begin{align}\label{Def: Outage_Classification_Capacity}
{C}_{\sf out} = \left\{ 
 \begin{array}{rcl}
 \max\limits_{r} \l\{ C^{\sf sel}(r) \ | \ P_{\sf out}(r) \leq \delta \r\}, &  & \text{class-selection case}是啊\\
  \max\limits_{r} \l\{ C^{\sf rnd}(r) \ | \ P_{\sf out}(r) \leq \delta \r\}, &  &  \text{random-class case},
\end{array}
\right. 
\end{align}
where $C^{\sf sel}(r)$ and $C^{\sf rnd}(r)$ are defined in~\eqref{Def: Classification_Capacity_sel} and~\eqref{Def: Classification_Capacity_rnd}, respectively. 

\begin{remark}
(Effective Classification Error Probability).
\emph{For remote classification, when the channel is in outage, the server receives zero features for a transmitted sample  but may make a random guess on the sample's label with the error probability of $\frac{L-1}{L}$. If this is the case, the effective classification error probability is  slightly larger than  its constraint $\epsilon$ and should be given as $(1-\delta)\epsilon + \delta \frac{L-1}{L}$.} 
\end{remark}

\section{Classification Capacity with Class Selection}
\label{sec:class_selection}

In this section, the $\epsilon$-classification capacity and its ergodic and outage counterparts for the class-selection case are analyzed. 

\subsection{Classification Error Probability}

To facilitate the derivation of classification capacities under a constraint on the classification error probability, we first analyze the probability as follows.

\subsubsection{Pairwise classification error probability}
Consider the classification of two specific classes, namely $\bU_i$ and $\bU_j$. The error probability of binary classification based on a similar data distribution model as the current one was studied in~\cite{Marzetta2000TIT} in the context of space-time demodulation. Let ``$i\rightarrow j$'' denote the event that a sample of class $i$ is assigned label $j$ by the classifier. Then the \emph{pairwise classification error probability} (PCEP) can be defined as $P(i\to j) = \Pr\l(\mathcal L(\bx) = j\ |y=i,  \mathcal{U}_L, R \r)$. A main result from~\cite{Marzetta2000TIT} is given below. 
\begin{lemma}\label{Lemma: PEP}
(Exact Pairwise Classification Error Probability~\cite{Marzetta2000TIT}).
\emph{
The probability is given as
\begin{equation}\label{eqn:exact_pep_v}
P(i\to j) = \frac{1}{4\pi}\int_{-\infty}^{\infty}dw \frac{1}{w^2+1/4}\cdot \prod_{\overset{k=1}{\cos\theta^{(i,j)}_k<1}}^{K}\frac{1+\sigma_{\bs}^2}{\sigma_{\bs}^4\left(1-\cos^2\theta^{(i,j)}_k\right)\left(\omega^2+a^2_k\right)},
\end{equation}
where $\theta^{(i,j)}_k$ denotes the $k$-th principal angle between $\bU_i$ and $\bU_j$, and $a_k=\sqrt{\frac{1}{4}+\frac{\sigma_{\bs}^2+1}{\sigma_{\bs}^4\left(1-\cos^2\theta^{(i,j)}_k\right)^2}}$.}
\end{lemma}
Note that the effect of the number of features per sample, $N$, (or the proportional communication rate, $R$) is not reflected in the above result given a fixed distance between $\bU_{i}$ and $\bU_{j}$, measured by $\{\cos^2{\theta_{k}^{(i,j)}}\}$. The effect of $N$ (or $R$) lies in determining the dimensionality of the feature space and hence the number of classes that can be packed into the space as elaborated in Section~\ref{sec: inst_class_capacity_selection}.
To simplify analysis and gain insight, we further derive an upper and a lower bounds on the probability in the following lemma. 

\begin{lemma}\label{Lemma:PCEPBound}
(PCEP Bounds).
\emph{The PCEP can be bounded as}
\emph{
\begin{equation}\label{Upper: PEP}
\frac{1}{\pi} \arctan(\sqrt{3})\l(\frac{1}{1+\frac{4}{K}g(\sigma^2_{\bs})d^2_{i,j}}\r)^K\leq P(i\to j) \leq \frac{1}{2}\left(\frac{1}{1+g(\sigma^2_{\bs})}\right)^{\lfloor{d^2_{i,j}}\rfloor},
\end{equation}
where $g(\sigma^2_{\bs})=\frac{1}{4\left(\sigma_{\bs}^{-4}+\sigma_{\bs}^{-2}\right)}$ is a monotonically increasing function of the data SNR $\sigma^2_{\bs}$ and $d_{i,j} = \sqrt{K-\text{tr}\{\bU_i\bU^T_i\bU_j\bU^T_j\}}$ denotes the (chordal) subspace distance between the two classes $\bU_i$ and $\bU_j$.}
\end{lemma}
\begin{proof}
See Appendix~\ref{Proof: Upper_bound}.
\end{proof}

\begin{remark}
(Effects of Data SNR and Class Distance).
\emph{
On the one hand, increasing the  data SNR causes data clusters to shrink, improving their discernibility. For this reason,   it is observed that both bounds on the PCEP in the above lemma decrease as the data SNR grows. On the other hand, the subspace distance between two classes determines their differentiability.  Consequently, increasing the distance reduces the bounds on the  PCEP. The improvement  is known as  the \emph{discrimination gain} in the literature. 
}
\end{remark}

\subsubsection{Classification error probability of $L$ classes}

Consider the error events $\{i\rightarrow j \ | \ i \neq j  \}$ and the pairwise classification error probability analyzed in the preceding subsection. By the \emph{union bound} and invoking~\eqref{Upper: PEP}, the probability can be bounded in terms of the pairwise counterpart as
\begin{equation}
P_{\sf e} = {\frac{1}{L} \sum_{i=1}^L \Pr \left( \bigcup_{j \neq i} (i\rightarrow j) \right) \stackrel{(a)}{\leq} \frac{1}{L} \sum_{i=1}^L \sum_{j \neq i} P(i\rightarrow j)   \stackrel{(b)}{=} } \frac{2}{L }\sum_{i=1}^{L-1}\sum_{j=i+1}^{L}P(i\rightarrow j),
\end{equation}
where (a) is due to the union bound and (b) due to the symmetry $P(i\rightarrow j) =P(j\rightarrow i)$.
Define $d_{\sf min} =\min_{(i,j)}d_{i,j}$, the above bound can be further relaxed to give the upper bound in Lemma~\ref{Lemma: PeBounds} in the sequel.

Next, $P_{\sf e} $ can be lower bounded as follows. Define $\mathcal W_{i,j^{*}}$ as an event that the ground-truth label is $i$ while the inferred label is $j^{*}\neq i$ subject to $d_{i,j^{*}}=  \min\limits_{j\neq i} d_{i,j}  \triangleq d^{(i)}_{\min}$. Then it follows from~\eqref{eqn: Pe_conditioned_on_R} that one lower bound of the classification error probability can be calculated as
\begin{equation}
P_{\sf e} \geq \frac{1}{L}\sum_{i=1}^{L}{\Pr\left(\mathcal W_{i,j^{*}} \ | \ i, \ j^{*} = \arg \min\limits_{j\neq i} d_{i,j} \right)},
\end{equation}
yielding the lower bound in Lemma~\ref{Lemma: PeBounds}.

\begin{lemma}\label{Lemma: PeBounds}
(Classification Error Probability).
\emph{Given the class subspace set $\mathcal{U}_L$, the classification error probability  can be bounded as
\begin{equation}\nn
\frac{1}{3L} \sum_{i=1}^{L} \l(\frac{1}{1+\frac{4}{K}g(\sigma^2_{\bs})\left(d^{(i)}_{\min}\right)^2}\r)^K \leq P_{\sf e}  \leq \frac{L}{2}\left(\frac{1}{1+g(\sigma_{\bs}^2)}\right)^{\lfloor{d^2_{\sf min}}\rfloor}. 
\end{equation}
}
\end{lemma}

\begin{remark}
(Effect of Number of Classes).
\emph{Apart from the effects of data SNR and class distance discussed earlier, one can further observe/infer from the above bounds that increasing  the number of classes,  $L$, makes the classification error probability grow. This is because that packing more classes into a fixed feature space reduces inter-class distances and thereby compromise their differentiability. }
\end{remark}

\subsection{$\epsilon$-Classification Capacity}
\label{sec: inst_class_capacity_selection}
The key step in deriving the $\epsilon$-classification capacity is to establish the equivalence between the classification capacity maximization via class selection and the Grassmannian packing problem.  To facilitate the analysis, we consider the scenario where the large-scale remote classifier at the server can support flexible classification as stated in the following assumption.

\begin{assumption}\label{Assump: library_distribution}
(Flexible Classification).
\emph{The server with a large class library supports classification of an arbitrary dataset (parameterized by a subspace set $
\mathcal U_{L}$) with the classification error probability $P_{\sf e}(R,\mathcal U_{L})$ in~\eqref{Def: Classification_Capacity_sel}.
}
\end{assumption}
In practice, large-scale classification realizes flexible classification  using a hierarchical architecture comprising a large number of component classifiers \cite{Yu2017SIGIR,Malerba2007JIIS}.

\subsubsection{Equivalence to Grassmannian packing}
Given $(N,K)$, a Grassmann manifold, $\mathcal{G}(N,K)$, refers to the space of $K$-dimensional subspaces embedded in the $N$-dimensional space, or equivalently the space of $N\times K$ unitary matrices.
Based on the definition of $\epsilon$-classification capacity in~\eqref{Def: Classification_Capacity_sel} and Assumption~\ref{Assump: library_distribution}, the subspace set $\mathcal U^*$ that represents the class selection for capacity maximization can be found by solving the following optimization problem:
\begin{equation}\label{opt_rrm}{\bf (P1)}\quad
\begin{aligned}
   \mathcal U^* = \arg \max_{ \mathcal U \in \mathcal{G}}\quad&C(\mathcal U)\\
    \text{s.t.}\quad\quad&P_{\sf e}(\mathcal U) \leq \epsilon, \
\end{aligned}\nn
\end{equation}
where $\mathcal{G} = \mathcal{G}(N,K)$, $C(\mathcal U) = C(R,\mathcal U)$ and $P_{\sf e}(\mathcal U) = P_{\sf e}(R,\mathcal U)$ with $N$, $K$, and $R$ in this subsection and omitted to simplify notation. Substituting the upper bound on $P_{\sf e}$ in Lemma~\ref{Lemma: PeBounds} into $(\bf P1)$, the problem can be recast as
\begin{equation}\label{opt_rrm1}{\bf (P2)}\quad
\begin{aligned}
   \mathcal U^* = \arg \max_{|\mathcal U| = L,\ \mathcal U \in \mathcal{G}}\quad&L\\
    \text{s.t.}\quad\quad&d_{\min} \geq \beta_L, \
\end{aligned}\nn
\end{equation}
where $\beta_L = \sqrt{\frac{\log_2\frac{L}{2\epsilon}}{\log_2(1+g(\sigma^2_{\bs}))}}$. The solution of $\bf (P2)$ lower bounds the maximum capacity from solving $\bf (P1)$ and the approximation is accurate when the error probability is small. An intuitive interpretation of Problem $(\bf P2)$ is to pack as many balls as possible (maximizing $L$), each centered at an element of $\mathcal U_L$ and with the radius  $\frac{d_{\min}}{2}$, into the space $\mathcal{G}$, giving the name \emph{Grassmannian packing} \cite{Conway1987Book}. A standard approach for solving this class of mathematical problems is to convert them into equivalent problems of maximizing the minimum separation distance among $L$ balls~\cite{Conway1987Book}:
\begin{equation}
\text{(Grassmannian Packing)}\quad  \mathcal U^*  = \arg\max_{\substack{\mathcal U\in \mathcal{G},\\ |\mathcal U|=L}} \ d_{\min}.
\end{equation}
Let $d^*_{\min}(L)$ denote the result from solving the above problem, called minimum class separation from packing. Then, $L$ is increased to reach the maximum value under the constraint $d^*_{\min}(L) \geq \beta_{L}$, thereby solving the original Problem $(\bf P1)$.

Though typically  Grassmannian  packing problems are intractable and usually solved numerically, there exists a rich literature on bounding the resultant minimum distance $d^{*}_{\min}(L)$. (see e.g.,~\cite{Conway1987Book}). The following particular result is from~\cite{Nogin2002TIT}.

\begin{lemma}\label{lemma:packing_bound}
(Packing Bounds).
\emph{
 For large feature-space  dimensions $N$, the minimum class separation distance from Grassmannian packing can be bounded as
\begin{equation}\label{Bounds_packing}
K L^{-\frac{2}{NK}} \lesssim [d^*_{\min}(L)]^2 \lesssim  2K\l[1-\l(1-L^{-\frac{2}{NK}}\r)^2\r], \qquad N\rightarrow\infty.
\end{equation}
}
\end{lemma}

\subsubsection{Packing bounds on $\epsilon$-classification capacity}

Using Lemmas~\ref{Lemma: PeBounds} and \ref{lemma:packing_bound},  the $\epsilon$-classification capacity defined in~\eqref{Def: Classification_Capacity_sel}   can be bounded as $C_{\sf lb}\leq C^{\sf sel}(R) \leq C_{\sf ub}$ with 
\begin{align}
C_{\sf lb}&=  \left\{L:   \frac{L}{2}\left(\frac{1}{1+g(\sigma_{\bs}^2)}\right)^{{d^2_{\sf lb}}}=  \epsilon \right\}\label{LB: C2_Capacity},\\ 
C_{\sf ub}&=  {\left\{L: \frac{1}{3} \l(\frac{1}{1+\frac{4}{K}g(\sigma^2_{\bs})d^2_{\sf ub}}\r)^K = \epsilon \right\},}\label{UB: C2_Capacity}
\end{align}
{where $d^2_{\sf lb} = K L^{-\frac{2}{NK}}$, $d^2_{\sf ub} = 2K\l(1-\l(1-L^{-\frac{2}{NK}}\r)^2\r)$ and~\eqref{UB: C2_Capacity} follows by substituting $\{d_{\sf{min}}^{(i)}\}$ in the lower bound of Lemma~\ref{Lemma: PeBounds} with $d_{\sf ub}$ as $[d_{\sf{min}}^{(i)}]^2 \leq K\leq d^2_{\sf ub}, \forall i$.} Note that $\lfloor\cdot\rfloor$ is omitted  as $\lfloor d^2_{\sf lb} \rfloor\rightarrow d^2_{\sf lb} $ for  large $N$ and $K$, which is a typical  case in multimedia  classification. Solving the two equations \eqref{LB: C2_Capacity} and \eqref{UB: C2_Capacity} and substituting $N =\beta R$ yield the following  theorem. 
\begin{theorem}\label{Theorem: Bounds_C2_Capacity}
($\epsilon$-Classification Capacity with Class Selection).
\emph{Consider the class-selection case. For a high communication rate, the capacity can be asymptotically bounded as:
\begin{equation}\label{Eq: bounds_classification_capacity}
2^{\frac{\beta R}{2}\l(K\log_2{K}+ Kc_{\sigma_{\bs}} - Kc_{\epsilon}\r)} \lesssim C^{\sf sel}(R) \lesssim {2^{\frac{\beta R}{2}\l(K\log_24K+K\log_2\frac{4g(\sigma^2_{\bs})}{1-3\epsilon}\r)}, \quad R\to \infty,}
\end{equation}
where $c_{\sigma^2_{\bs}} = \log_2\log_2(1+ g(\sigma^2_{\bs}))$ and $c_{\epsilon} = \log_2 \log_2\frac{1+g(\sigma^2_{\bs})}{2\epsilon}$. In particular, as $R, K \to \infty$, the capacity scales as
\begin{equation}\label{Eq: scaling_classification_capacity_class-selection}
\lim_{R,K\to \infty} \frac{\log_2C^{\sf sel}(R,K)}{RK\log_2K} = \frac{\beta}{2}.
\end{equation}
}
\end{theorem}
\proof
See Appendix~\ref{Proof: Bounds_C2_Capacity}.
\endproof

\begin{remark}
(Mathematical Intuition for  Capacity Scaling Laws).
\emph{One can observe from the above theorem that the $\epsilon$-classification capacity increases \emph{exponentially}  as the instantaneous communication rate $R$ grows. The underpinning mathematical reason is that  the volume of the Grassmann manifold containing the dataset classes 
is an \emph{exponential function} of its dimensions $N$, which  is  proportional to $R$. Consequently, increasing $R$ allows an exponentially growing number of ``balls" (classes) to be packed into the manifold.  One the other hand, the  capacity scales \emph{super-exponentially} with the dimensions of each data cluster (or each class), namely $K$. Note that increasing $K$ improves the inter-class differentiability. One can infer from \eqref{LB: C2_Capacity} and~\eqref{UB: C2_Capacity} that with the classification error probability fixed,  the allowed number of ``balls" ($L$)  grows \emph{exponentially} as the  minimum pairwise distance of the ``balls" (classes), $d^*_{\min}$, increases. Furthermore, $d^*_{\min}$ is a \emph{super-linear} function of $K$ as one can further observe from  the definitions of $d^2_{\sf ub}$ and $d^2_{\sf lb}$ after \eqref{UB: C2_Capacity}. Combining the two relations gives the super-exponential capacity  scaling w.r.t. $K$. 
}
\end{remark}

\begin{remark}
(Effects of QoE Requirement and Data/Transmit SNR).
\emph{
The dependence of the $\epsilon$-classification capacity on the allowed maximum classification error probability $\epsilon$ (or QoE requirement), the data SNR $\sigma^2_{\bs}$, and the transmit SNR $\rho$ can be interpreted geometrically in terms of Grassmannian packing. Increasing $\epsilon$, $\sigma^2_{\bs}$ and $\rho$ allows ``balls" (classes) to get closer, shrinks ``ball radiuses" (the variance of each data cluster), and increasing the Grassmannian volume (communication rate), respectively. They  all  contribute to packing more ``balls" (larger capacity) though in different ways. 
}

\end{remark}

\subsection{Ergodic and Outage Classification Capacities}\label{Rayleigh_assumption}
Given a distribution function of the communication rate $R$, it is straightforward to use the results in Theorem~\ref{Theorem: Bounds_C2_Capacity} to analyze the ergodic and outage classification capacities based on their definitions in~\eqref{eqn: ergodic_capacity_definition} and~\eqref{Def: Outage_Classification_Capacity}. In this section, we consider a Rayleigh fading channel and perform such analysis to provide concrete insight into the effect of channel fading on the  performance of remote classification. 

\subsubsection{Ergodic classification capacity}
Correspondently, the ergodic channel capacity is $\bar{R} = \mathbb E[B\log_2(1+ \rho|h|^2)]$, where $ \rho$ is the transmit SNR and  the channel gain  $|h|^2=\exp(1)$.

\begin{proposition}\label{Proposition: bounds_Ergodic_capacity}
(Ergodic Classification Capacity for Rayleigh Fading). 
\emph{Consider the class-selection case. The ergodic classification capacity defined in~\eqref{eqn: ergodic_capacity_definition} can be bounded as 
\begin{equation}\label{Eq: Lower_bound_ergodic_capacity}
\sqrt{2\pi\gamma_{\sf lb}}\cdot  \rho^{\gamma_{\sf lb}} \cdot e^{\gamma_{\sf lb}(\log \gamma_{\sf lb}-1)} \leq \bar{C} \leq \sqrt{2\pi\gamma_{\sf ub}} \cdot \rho^{\gamma_{\sf ub}} \cdot e^{\gamma_{\sf ub}(\log \gamma_{\sf ub}-1)}, 
\end{equation}
where $\gamma_{\sf lb} = \frac{\beta B}{2}\l(K\log_2{K}+ Kc_{\sigma_{\bs}} - Kc_{\epsilon}\r)$ and ${\gamma_{\sf ub} = \frac{\beta B}{2}\l(K\log_24K+ K\log_2\frac{4g(\sigma^2_{\bs})}{1-3\epsilon}\r)}$ with $c_{\sigma_{\bs}}$ and $c_{\epsilon}$ defined in Theorem~\ref{Theorem: Bounds_C2_Capacity}. In particular, for large $\bar{R}$ and $K$, the capacity scales as 
\begin{equation}\label{Eq: scaling_ergodic_classification_capacity_class-selection}
\lim_{\bar{R}, K\to \infty} \frac{\log_2\bar{C}}{\bar{R}K\log_2K} = \frac{\beta}{2},
\end{equation}
where $\beta = \frac{T}{Q}$.}
\end{proposition}
\proof
See Appendix~\ref{Proof: Bounds_Ergodic_Capacity}.
\endproof

\begin{remark}
(Fading Does Not Affect Capacity Scaling).\label{Re:FadingEffect}
\emph{
The key observation  from the above proposition is that both the scaling laws of the ergodic classification capacity are the same as those for $\epsilon$-classification capacity in Theorem~\ref{Theorem: Bounds_C2_Capacity} except for replacing the instantaneous rate $R$ with its ergodic counterpart $\bar{R}$. The remark also applies to outage classification capacity analyzed in the sequel if the communication rate is  modified as the maximum rate under an outage constraint.
}
\end{remark}

\subsubsection{Outage classification capacity}\label{Subsection: outage_classification_capacity_fixed} To begin with, the maximum communication rate, denoted as  $R_{\delta}$, can be obtained from the active outage constraint  $P_{\sf{out}} = \Pr(R\leq R_{\delta}) \leq \delta$ and the exponential distribution of the channel gain as
\begin{equation}\label{Eq: outage}
R_{\delta} = B\log_2\l(1 + \rho\log\l(\frac{1}{1-\delta}\r)\r).
\end{equation}   
It is worth mentioning that  $R_{\delta}$ is a monotonically increasing function of the outage probability $\delta$. Moreover,  note that the corresponding number of transmitted features per sample is now given as  $N = \beta R_{\delta}$. The outage classification capacity is equal to the $\epsilon$-classification capacity by replacing $R$ with $R_{\delta}$ in \eqref{Eq: outage}, yielding the following proposition. 

\begin{proposition}\label{Proposition: outage_classification_capacity}
(Outage Classification Capacity for Rayleigh Fading)
\emph{
Consider the class-selection case. The outage  classification capacity defined in~\eqref{Def: Outage_Classification_Capacity} can be bounded as 
\begin{equation}\label{eqn: C2capacity_LB_fixed_model}
 \left[1+\rho\log\left(\frac{1}{1-\delta}\right)\right]^{\gamma_{\sf{lb}} }\;  \leq \; \ C_{\sf{out}} \leq \left[1+\rho\log\left(\frac{1}{1-\delta}\right)\right]^{\gamma_{\sf{ub}}},
\end{equation}
with $\rho \gg 1$, ${\gamma_{\sf{lb}} }$ and ${\gamma_{\sf{ub}} }$ defined in Proposition~\ref{Proposition: bounds_Ergodic_capacity}. In particular, as $R_{\delta}, K\to \infty$, the capacity scales as 
\begin{equation}\label{Eq: scaling_outage_classification_capacity_class-selection}
\lim_{{R}_{\delta},K\to \infty} \frac{\log_2{C}_{\sf out}}{{R}_{\delta}K\log_2K} = \frac{\beta}{2}.
\end{equation}
}
\end{proposition}
\proof
See Appendix~\ref{Proof: scaling_outage_Capacity}.
\endproof

\section{Classification Capacity with Random Classes}
\label{sec:Pe_random_data}

In this section, the $\epsilon$-classification capacity and its ergodic and outage counterparts are analyzed for the random-class case and compared with their counterparts in the class-selection case.

\subsection{Expected Classification Error Probability}
The  expected classification error probability is analyzed  in this subsection for  a dataset  with i.i.d. isotropic classes, $\{\mathbf{U}_\ell\}$, on the Grassmannian $\mathcal{G}(N,K)$.

\subsubsection{Distribution of class  separation}
Let $\theta_{\max}$ denote the maximum principal angle between a pair of  classes, $\mathbf{U}_i$ and $\mathbf{U}_j$.

\begin{lemma}\label{statistic_information}
(Class Separation Distribution~\cite{b5}).
\emph{The PDF of $X = \sin^2\theta_{\max}$ is given as 
\begin{equation}
\label{eqn:PDF_of_sin_square_theta_subspace}
f_X(x)=c_{N,K,\theta_{\max}} \ \tensor*[_2]{F}{_1}\l(\frac{N-K-1}{2}, \frac{1}{2};\frac{N + 1}{2};\sin^2\theta_{\max}{\bf I}_{K-1}\r),
\end{equation}
where $c_{N,K,\theta_{\max}}  =K(N-K)\frac{\Gamma\left(\frac{K+1}{2}\right)\Gamma\left(\frac{N-K+1}{2}\right)}{\sqrt{\pi}\Gamma\left(\frac{N+1}{2}\right)}(\sin \theta_{\max})^{K(N-K)-1}$ and $\tensor*[_2]{F}{_1}(\cdot)$ denotes the Gaussian hypergeometric function  with a matrix argument. 
}
\end{lemma}
The squared chordal distance between $\bU_i$ and $\bU_j$ is defined as  $d^2_c(\bU_i,\bU_j) = K-\text{tr}\{\bU_i\bU_i^T\bU_j\bU_j^T\}$. Using Lemma~\ref{statistic_information}, an upper  bound on the \emph{cumulative distribution function} (CDF) of the  distance is derived as shown in the lemma below.
\begin{lemma}\label{lemma:CDF_subspace_disntance_bounds}
(Upper Bound on Class Separation Distribution). \emph{Consider a pair of  independent and isotropic classes $\bU_i$ and $\bU_j$ on the Grassmannian $\mathcal{G}(N,K)$.   The CDF of their squared chordal distance $d^2_c(\bU_i,\bU_j)$, denoted as $F_{d_{c}^2}(x)$,  can be bounded as
\begin{equation}\label{eqn:CDF_subspace_distance_UB}
F_{d_{c}^2}(x)\leq \left(\frac{x}{K}\right)^{\frac{K(N-K)}{2}}, \ x\in[0,K].
\end{equation}
}
\end{lemma}
\begin{proof}
See Appendix~\ref{sec:appendix_proof_CDF_subspace_disntance_bounds}.
\end{proof}

\subsubsection{Expected classification error probability} Consider the ML classification of two random classes. Using Lemmas~\ref{Lemma:PCEPBound} and \ref{lemma:CDF_subspace_disntance_bounds},  the expected PCEP  can be bounded as follows.

\begin{lemma}\label{lemma:Expectation_PEP_bounds}
(Upper Bound on Expected PCEP).
\emph{
For a pair of independent and isotropic random classes $\bU_i$ and $\bU_j$,  the expected PCEP can be upper-bounded as
\begin{equation}
\mathbb E[P(i\to j)] \leq \frac{1}{2}\left(\frac{1}{1+g(\sigma_{\bs}^2)}\right)^{K}+\frac{\log\left(1+g(\sigma_{\bs}^2)\right)}{\left(1+g(\sigma_{\bs}^2)\right)}\frac{1}{N}.
\label{inqe:avg_PEP_UB}
\end{equation}
}
\end{lemma}

\begin{proof}
See Appendix~\ref{sec:appendix_proof_Expectation_PEP_bounds}.
\end{proof}

 By applying the union bound and using  Lemma~\ref{lemma:Expectation_PEP_bounds}, we obtain the following lemma. 
 
 \begin{lemma}\label{lemma:Expectation_Pe}
(Expected Classification Error Probability).
\emph{
For  a dataset having $L$  independent and isotropic  classes $\mathcal{U}_L = \{\bU_\ell\}$,  the expected classification error probability  can be upper-bounded as
\begin{equation}
\label{UB: unionBound_Pe_randomData}
P^{\sf rnd}_{\sf e}(L, R) = \mathbb E_{\mathcal{U}_L}[ P_{\sf e}( \mathcal{U}_L, R)]\leq \frac{L}{2}\left[\frac{1}{2}\left(\frac{1}{1+g(\sigma_{\bs}^2)}\right)^{K}+\frac{\log\left(1+g(\sigma_{\bs}^2)\right)}{\left(1+g(\sigma_{\bs}^2)\right)}\frac{1}{N}\right],
\end{equation} 
where $N = \beta R$. 
}
\end{lemma}

\subsection{$\epsilon$-Classification Capacity}
The $\epsilon$-classification capacity defined in~\eqref{Def: Classification_Capacity_rnd} can be obtained by solving:
\begin{equation}\label{opt_rrm2}{\bf (P3)}\quad
\begin{aligned}
   C^{\sf rnd}(R) = \arg \max_{L}\quad& L\\
    \text{s.t.}\quad\quad& P^{\sf rnd}_{\sf e}(L, R) \leq \epsilon.
\end{aligned}\nn
\end{equation}
By modifying the constraint using~\eqref{UB: unionBound_Pe_randomData}, the capacity can be lower-bounded as follows. 
\begin{theorem}\label{Theorem: Bounds_Classification_Capacity_random}
($\epsilon$-Classification Capacity with Random Classes).
\emph{For a large communication rate, the $\epsilon$-classification capacity for the random-class case  can be asymptotically bounded as:
\begin{equation}\label{Eq: Lower_bound_ergodic_capacity_random}
 C^{\sf rnd}(R) \gtrsim \frac{2\beta\epsilon \left(1+g\left(\sigma_{\bs}^2\right)\right)}{\log\left(1+g\left(\sigma_{\bs}^2\right)\right)}R, \qquad R\rightarrow \infty.
\end{equation}
}
\end{theorem}

\begin{remark}
(Mathematical Intuition for  Capacity Scaling Laws).
\label{Re:CapacityRandom}
\emph{As opposed to the \emph{exponential capacity} scaling for the class-selection case, the $\epsilon$-classification  capacity is shown in Theorem~\ref{Theorem: Bounds_Classification_Capacity_random} to scale \emph{linearly} w.r.t. the communication rate. Unlike deterministic classes resulting from Grassmannian packing in the former case, the random classes in the current case do not have a guaranteed minimum separation distance and the randomness in their  separations dramatically increases the classification error probability. As a result, the number of classes that can be contained in the Grassmannian has to be smaller so as to satisfy  a constraint on the expected separation distances, which determines  the expected classification error probability. This is the fundamental reason for much slower (linear) capacity scaling w.r.t. the  communication rate that determines the Grassmannian volume. On the other hand, the data-cluster dimensions $K$ does not appear in the scaling law as its effects on the classification error probability is negligible in the current case. This fact is reflected in the upper bound on the probability in Lemma~\ref{lemma:Expectation_PEP_bounds}  where  the second term independent of $K$  dominates  the first that varnishes exponentially fast as  $K$ increases.}
\end{remark}

\subsection{Ergodic and Outage Classification Capacities}
The linear scaling of the $\epsilon$-classification capacity w.r.t. to the communication rate $R$ makes it straightforward to extend the result to  ergodic and outage classification capacities by modifying $R$ accordingly, giving the following proposition. 

\begin{proposition}\label{Proposition: bounds_Ergodic_capacity_random}
(Ergodic and Outage Classification Capacities with Random Classes). 
\emph{The ergodic and outage classification capacities for the random-class case can be bounded as 
\begin{align}
  \bar{C}^{\sf rand} &\gtrsim \frac{2\beta\epsilon \left(1+g\left(\sigma_{\bs}^2\right)\right)}{\log\left(1+g\left(\sigma_{\bs}^2\right)\right)}\bar{R}, \quad \bar{R} \to \infty,\\
  {C}^{\sf{rand}}_{\sf out}&\gtrsim \frac{2\beta\epsilon \left(1+g\left(\sigma_{\bs}^2\right)\right)}{\log\left(1+g\left(\sigma_{\bs}^2\right)\right)}R_{\delta}, \quad R_{\delta} \to \infty,
\end{align} 
where $\bar{R}$ is the expected communication rate and $R_{\delta}$ the maximum rate under the outage constraint. 
}
\end{proposition}
A similar remark as Remark~\ref{Re:FadingEffect} can be made that fading affects the communication rate but does not change the capacity scaling laws w.r.t. to the rate, which are determined by the distribution of classes on the Grassmannian (see Remark~\ref{Re:CapacityRandom}).

\section{Extension to Fast Fading}
\label{Section:Fast_fading}

The preceding analysis assuming a static channel within each slot of transmitting a feature vector is extended to the case of channel variation within the slot due to fast fading. To this end, we modify the transmission and channel models as follows while other models and assumptions remain unchanged. To model fast fading, each slot is divided into sub-slots, over which the channel follows i.i.d. block fading. Considering an arbitrary slot, let $N$ features to be transmitted over the slot be divided into $S$ packets with $1\leq S \leq N$; each is transmitted using a sub-slot with a packet-loss probability (or equivalently outage probability) of $P_{\sf out} = \eta$. The features are extracted from the received packets and assembled as a single feature vector with missing features replaced by zeros, which is then used for classification. The variable $S$ is suitably called \emph{the fading speed}. {Consider the class-selection case where classes are packed on a Grassmannian embedded in the feature space. If the fraction of lost feature dimension is small, the classes constituting packing in the original space remains approximately so in the reduced-dimension space. Assuming such a case, the $\epsilon$-classification capacity is determined by the dimensionality of the latter space, or equivalently the number of successfully received features per sample, denoted as $N_x$. This also holds in the random-class case for a different reason that random erasures of some dimensions of the feature space does not change the isotropic distribution in the resultant space. The random variable} $N_x$ is determined by the number of successfully received packets, $X$, that follows the binomial distribution:
\begin{equation}\label{eqn: bionomial_num_features}
\Pr\l(X\r) = 
\binom{S}{n}(1-\eta)^{n}\eta^{S-n}.
\end{equation}
Given the average number of successfully received packets, $(1-\eta)S$, fixed, for large $S$, the distribution can be approximated as Poisson: 
\begin{equation}
\Pr(X = n) \approx  \frac{[(1-\eta)S]^ne^{-(1-\eta)S}}{n!},\qquad S \gg 1. 
\end{equation}
\subsubsection{Class-selection case}
Combining the approximate distribution function, $N_{x}=\frac{nN}{S}$ and Theorem~\ref{Theorem: Bounds_C2_Capacity}, the ergodic classification capacity is derived as 
\begin{equation}\label{Ineq: bouds_fast_fading}
 e^{(1-\eta)S\l(2^{\frac{\gamma_{\sf lb}}{\beta B}\frac{N}{S}}-1\r)}\leq \bar{C}\leq  e^{(1-\eta)S\l(2^{\frac{\gamma_{\sf ub}}{\beta B}\frac{N}{S}}-1\r)}, \ \ S\gg 1, \eta \ll 1. 
\end{equation}
Define the ergodic communication rate $\bar{R}=\frac{(1-\eta)N}{\beta}$. For the maximum fading speed $S=N$, the ergodic classification capacity scales as
\begin{equation}
\label{eq: scaling_fastErgodic}
 2^{\frac{\gamma_{\sf{lb}}}{\beta B}}-1 \leq \lim\limits_{  \bar{R} \rightarrow \infty}\frac{\log \bar{C}}{\beta \bar{R}} \leq 2^{\frac{\gamma_{\sf{ub}}}{\beta B}}-1 .
\end{equation}

The above results suggest the following. First, as the number of packets $S$ grows, both lower and upper bounds in~\eqref{Ineq: bouds_fast_fading} decrease, reflecting the effect of fast fading.  Next, the capacity scaling law in~\eqref{eq: scaling_fastErgodic} is exponential w.r.t. the ergodic communication rate as its slow-fading counterpart in Proposition~\ref{Proposition: bounds_Ergodic_capacity}. Therefore, the fading speed does not affect the classification-communication-rate  relation, {which is fundamentally attributed to class selection}, except for scaling the communication rate by the packet-success probability  $(1-\eta)$.

\subsubsection{Random-class case} The ergodic classification capacity in this case can be easily modified from its slow-fading counterpart in Proposition~\ref{Proposition: bounds_Ergodic_capacity_random} by redefining  the ergodic communication rate for the current case: 
\begin{equation}\label{Eq: scaling_ergodic_classification_capacity_class-random}
\bar{C}^{\sf rand} \gtrsim \frac{2\beta\epsilon \left(1+g\left(\sigma_{\bs}^2\right)\right)}{\log\left(1+g\left(\sigma_{\bs}^2\right)\right)}\bar{R}, \quad \bar{R} \to \infty, 
\end{equation}
where $\bar{R}=\frac{(1-\eta)N}{\beta}$. 
As before, the effect of fast fading is to scale the ergodic classification capacity by the packet-success probability $(1-\eta)$.

\section{Experimental Results}
\label{sec:experiments} 
\subsection{Experimental Settings}

Two sets of experimental results are obtained based on the statistical data model used in the preceding analysis and a real dataset, respectively. Their corresponding experiment settings are as follows. For all experiments, fading is modeled as Rayleigh, the transmit SNR is set as $15$ dB and channel bandwidth as $50$ KHz.
\begin{itemize}
\item \emph{Statistical data model:} The selected Grassmannian packing datasets were generated by Conway and Sloane~\cite{ConwayPack}. The maximum classification error probability is $0.03$ and $0.19$ for the class-selection and random-class cases, respectively, and the maximum (channel) outage probability is $0.3$. The data SNR is set as $15$ dB.
\item \emph{Real dataset:} The well known MNIST dataset is used that comprises images of handwritten numbers. {For inference, the popular neural network model, \emph{multi-layer perception} (MLP), is adopted as the classifier and trained using the training dataset of MNIST.} The maximum classification error probability is set as $0.02$.
\end{itemize}

\begin{figure*}[tt]
  \centering
  \subfigure[Class-Selection Case]{\label{Fig:a}\includegraphics[width=0.45\textwidth]{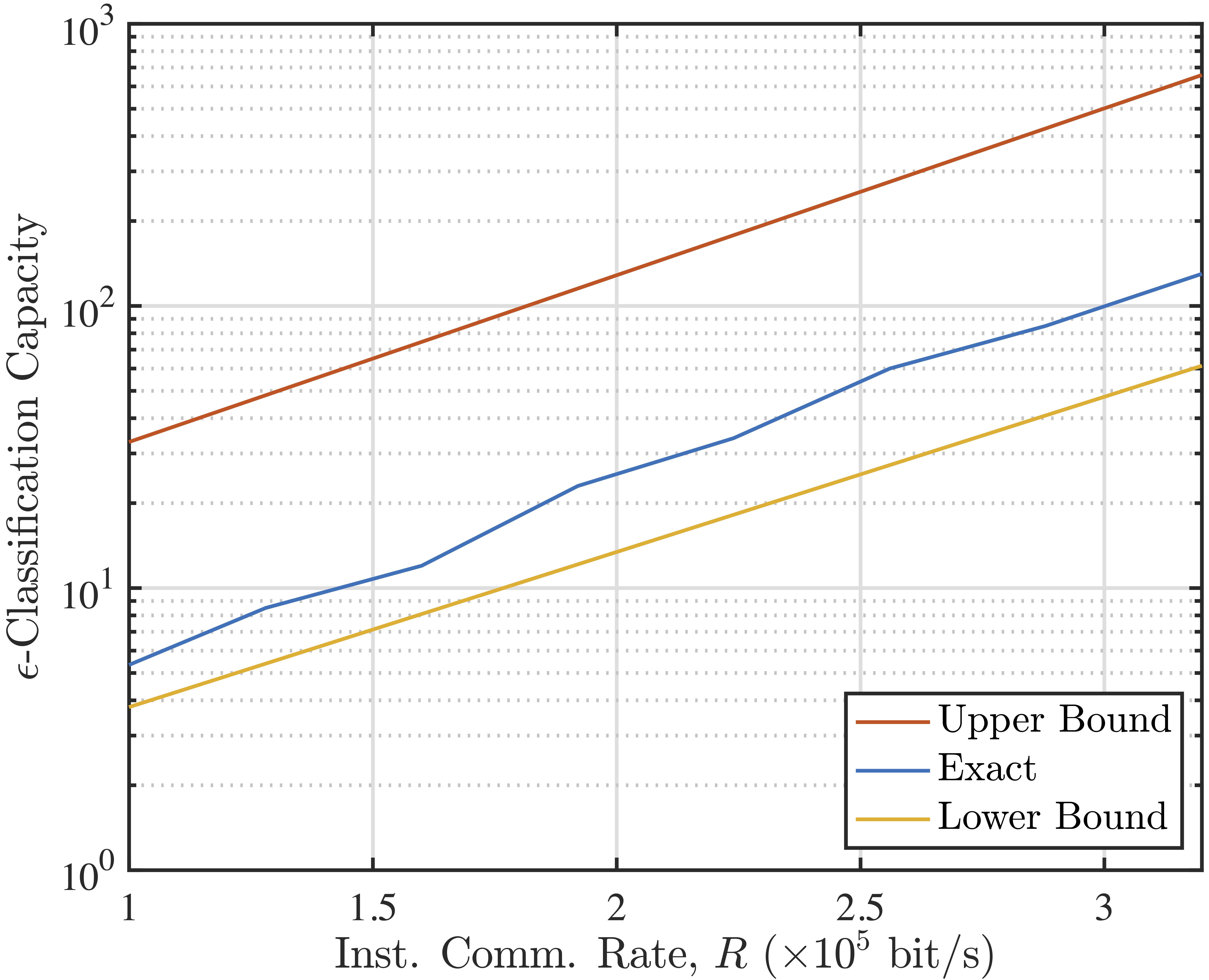}}
    \hspace{0.35in}
  \subfigure[Random-Class Case]{\label{Fig:b}\includegraphics[width=0.45\textwidth]{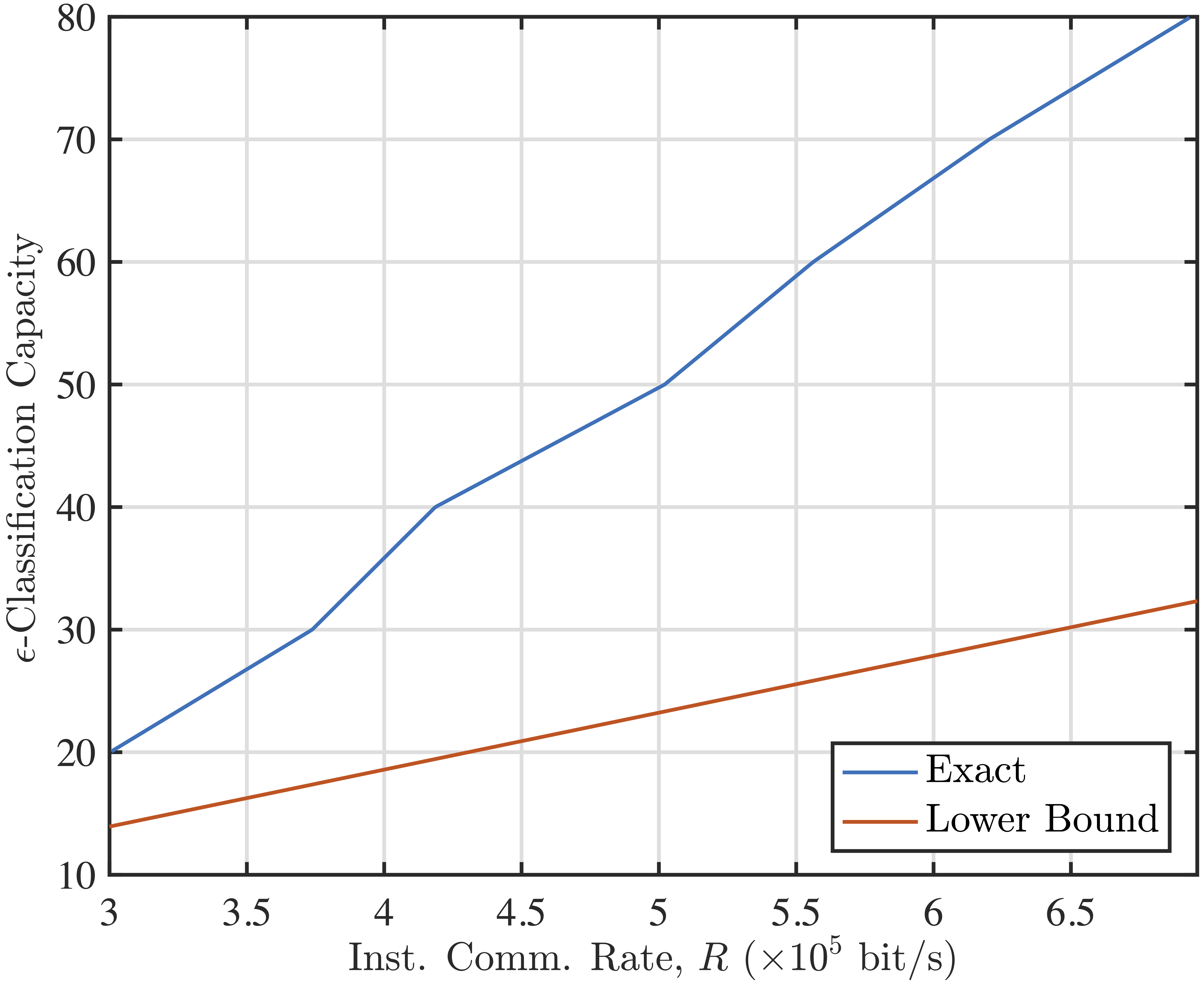}}
  \subfigure[Capacity Comparison for Class-Selection Case]{\label{Fig:c}\includegraphics[width=0.453\textwidth]{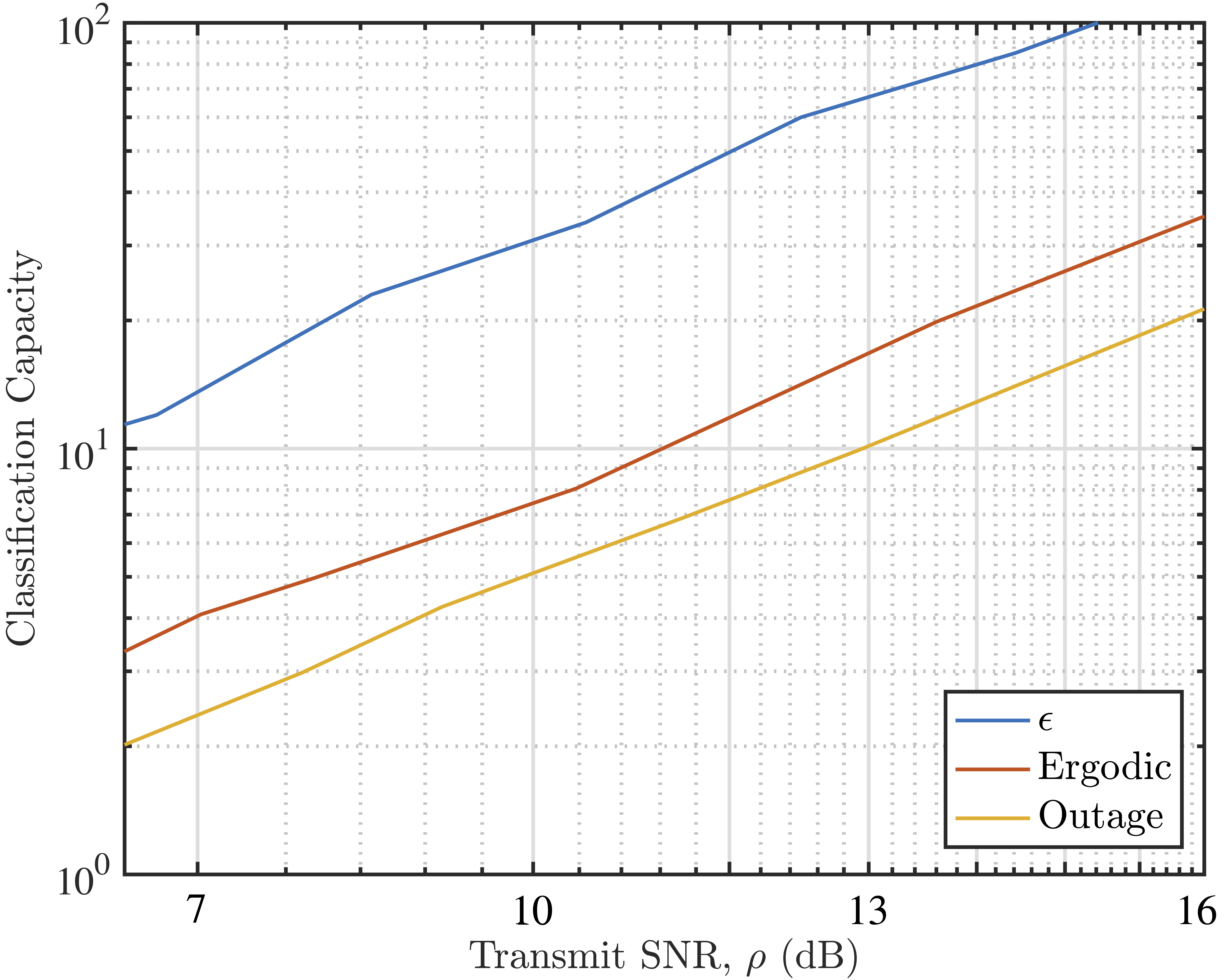}}
    \hspace{0.35in}
      \subfigure[Capacity Comparison for Random-Class Case]{\label{Fig:d}\includegraphics[width=0.445\textwidth]{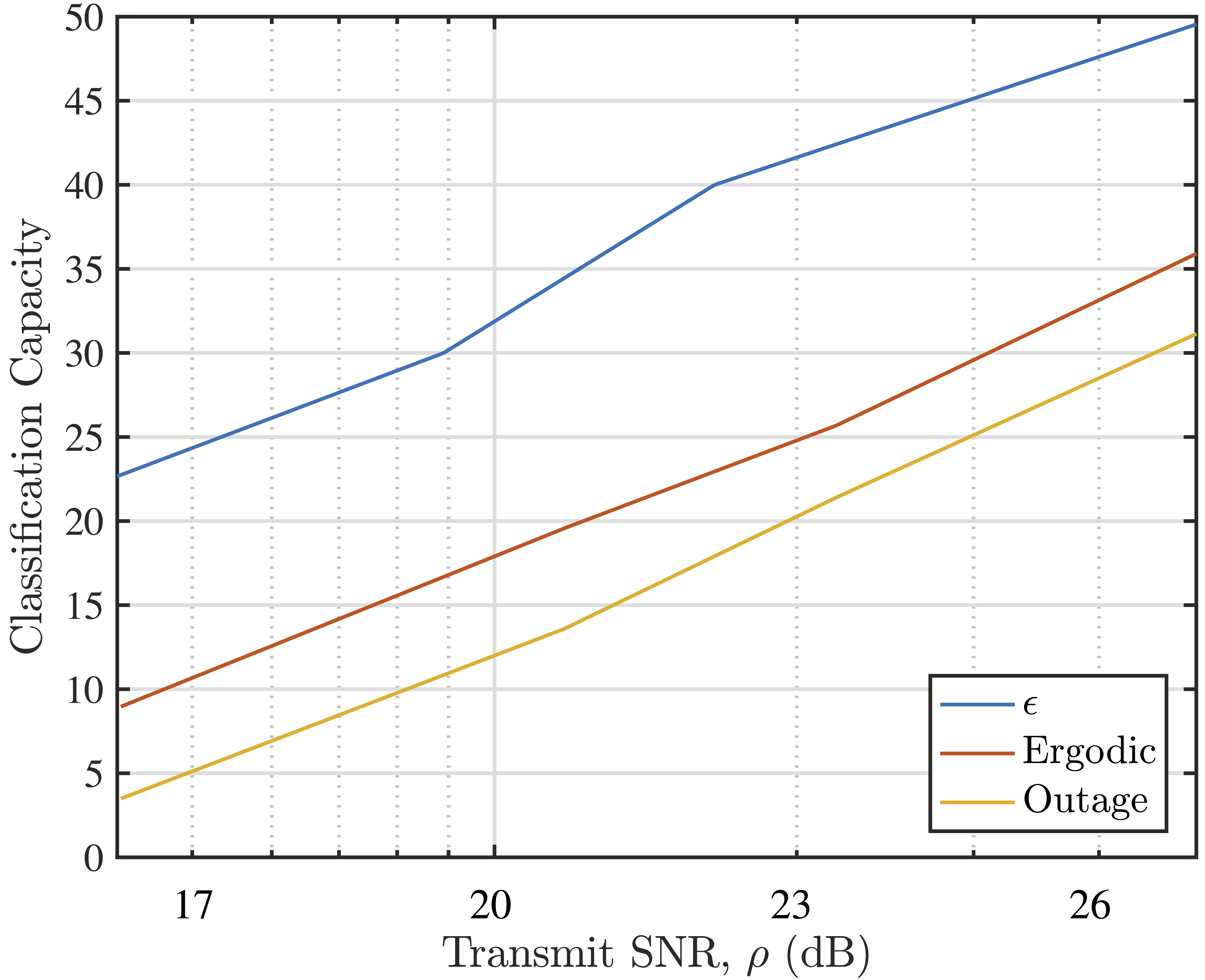}}
      \caption{Comparison of $\epsilon$-classification capacity, ergodic and outage classification capacities in both the channel-selection and random-class cases. }
        \label{fig_dfs_parameter}
  \vspace{-3mm}
\end{figure*}

%\vspace{-3mm}
\subsection{Classification Capacities with Statistical Data Model}
%\vspace{-2mm}
Fig. 2 shows the scalings of classification capacities of a remote-classification system as the communication rates grow and compares different capacity measures as well as the cases of class selection and random classes. The exponential and linear scaling laws of the $\epsilon$-classification capacities as presented in Theorems~\ref{Theorem: Bounds_C2_Capacity} and~\ref{Theorem: Bounds_Classification_Capacity_random} are shown in Fig. 2(a) and (b) to hold even in a practical regime. Note that the small duration of the curves is caused by numerical computation of Grassmannian packing~\cite{ConwayPack}. On the other hand, despite following the correct scaling laws, the bounds on the capacities are not tight due to the combined effect of the looseness of the union bounds on classification error probabilities and the distance bounds related to Grassmannian packing (in the class-selection case). Similar observations can be made on the bounds on ergodic and outage capacities with relevant curves omitted in Fig. 2 to keep the figures simple. Next, one can draw a conclusion from the comparisons in Fig. 2(c) and (d) that channel fading has a significant effect on the capacity of remote classification. For example, for a transmit SNR of $7$ dB, the ergodic capacity (with fading and CSIT) and outage capacity   (with fading but no CSIT) are $74 \%$ and $84\%$ less than the $\epsilon$-classification capacity (without fading), respectively, in the class-selection case; with a transmit SNR of $17$ dB, the losses are $64\%$ and $87\%$ in the random-class case. Last, comparing Fig. 2(a) and (b) reveals a substantial capacity gain due to class selection such as {$4$-time} increase in $\epsilon$-classification capacity at the {communication} rate of $3\times 10^5$ bit/s. The same conclusion holds for other capacity measures by comparing Fig. 2(c) and (d).

%\vspace{-3mm}
\subsection{Classification Capacities with Real Dataset}
%\vspace{-2mm}
Available class subsets are generated by different combinations of classes in the MNIST dataset. The example of three 3-class subsets is illustrated in Fig. 3.
\begin{figure}[h]
\centering
\includegraphics[scale=0.33]{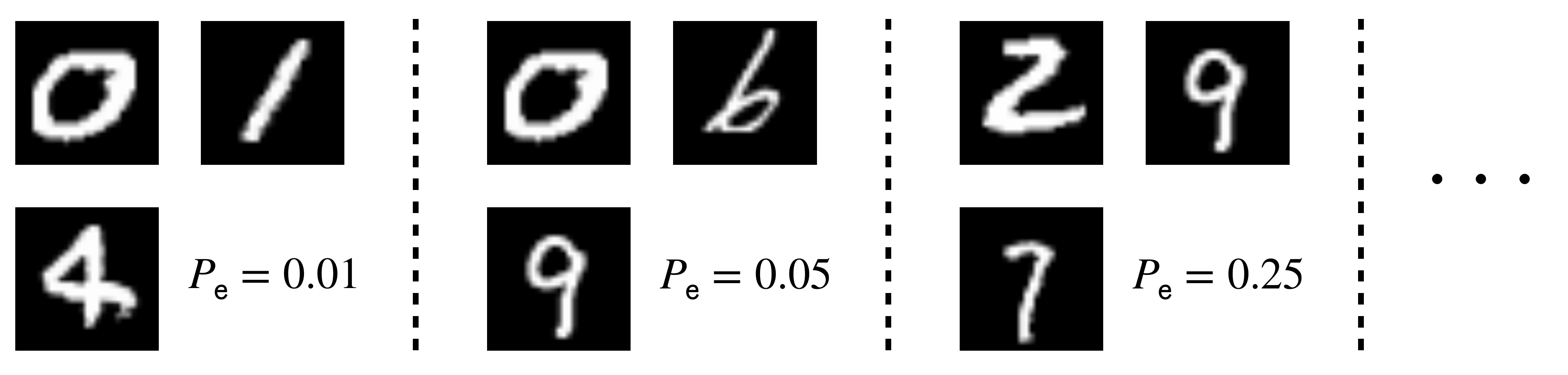}
\label{fig: mnist}
\caption{Example of $3$-class subsets of the MNIST dataset. Their corresponding classification error probabilities are different as specified.}
\end{figure}

The $\epsilon$-classification capacities for the cases of class selection and random classes are compared in Fig. 4. As the dataset is generated by the nature, the selected class subset is no longer generated by Grassmannian packing or the isotropic distribution as assumed in the theoretical analysis. However, we can still observe the capacity gain of class selection from Fig. 4, e.g., $33\%$ capacity gain at the communication rate of $10^{5}$ bit/s. Furthermore, the capacity with class selection scales with a growing communication rate at a rate faster than the random-class case. Both trends are aligned with the analytical results.

\begin{figure}[tt]
\centering
\includegraphics[scale=0.7]{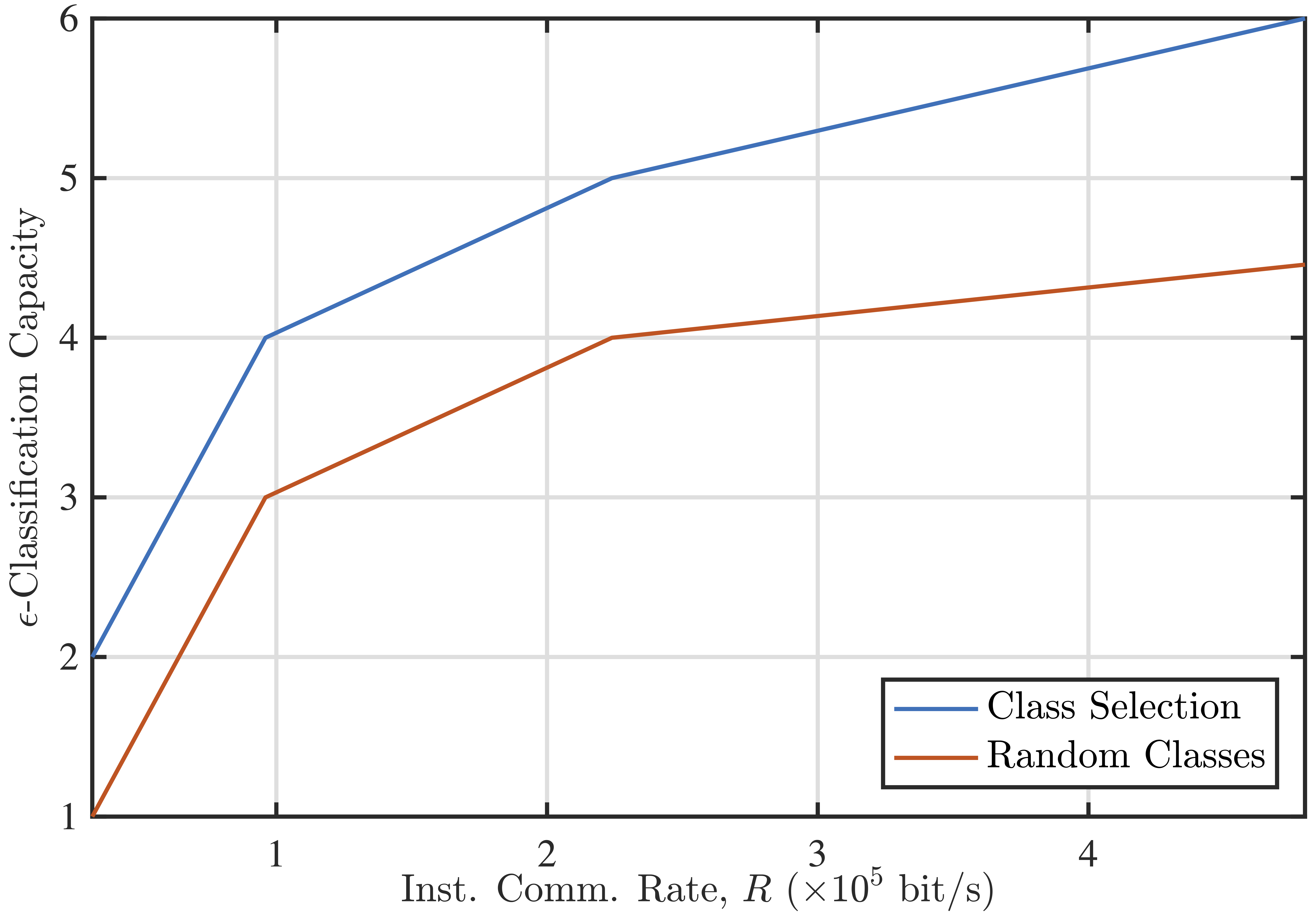}
\caption{Classification capacity comparison between the cases of class-selection and random-class with the MNIST dataset.}
\end{figure}

\section{Concluding Remarks}\label{Sec:Concluding Remarks}

In this work, we have studied the performance of remote classification over wireless channels. The main contribution is the establishment of a relation between classification and communication by proposing various metrics of classification capacities and analyzing them using tools from differential geometry. This has led us to discover that the freedom of choosing object classes for classification under the channel constraint can attain an exponential scaling law of classification capacity w.r.t. the communication rate; without {a deliberate selection}, the scaling is linear.

The current study opens numerous directions for further investigation. Several of them are particularly interesting, including {a realistic latency model, use }of advanced wireless techniques (e.g., MIMO and OFDM) to increase the classification capacity, {as well as} the design of multiuser remote classification system {that gives rise to new issues in terms of, e.g.,} resource allocation and cooperation.

\vspace{-3mm}
\appendix

\subsection{Proof of Lemma~\ref{Lemma:PCEPBound}}\label{Proof: Upper_bound}
First, we prove the upper bound on $P(i\to j)$. As $\omega^2 + a^2_k\geq a^2_k, \forall \omega $, it follows from~\eqref{eqn:exact_pep_v} that
\begin{align}\label{Upper: intermediate}
P(i\to j) &\leq \frac{1}{4\pi}\int_{-\infty}^{\infty}dw \frac{1}{w^2+1/4}\cdot \prod_{\overset{k=1}{\cos\theta^{(i,j)}_k<1}}^{K}\frac{1+\sigma_{\bs}^2}{\sigma_{\bs}^4a^2_k\left(1-\cos^2\theta^{(i,j)}_k\right)},\nn\\
& \leq \frac{1}{2}\prod_{k=1}^{K}\frac{1}{1+g(\sigma_{\bs}^2)\sin^{2}\theta_k^{(i,j)}}  \triangleq P_{\sf ub}.
\end{align}
On the other hand, one can easily verify that 
\begin{equation}
\frac{\partial P_{\sf ub}}{\partial \sin^{2}\theta_k^{(i,j)}}<0 \quad \text{and}\quad \frac{\partial^2 P_{\sf ub}}{\partial \left(\sin^{2}\theta_k^{(i,j)}\right)^2}<0.
\end{equation}
The above results suggest that, given $d^2_{i,j}=\sum_{k=1}^{K}\sin^2\theta^{(i,j)}_k$, $P_{\sf ub}$ is maximized when as many principal angles as possible are equal to zero. Consequently, one can further bound~\eqref{Upper: intermediate} as
\begin{equation}
P(i\to j) \leq \frac{1}{2}\left(\frac{1}{1+g(\sigma^2_{\bs})}\right)^{\lfloor{d^2_{i,j}}\rfloor}.
\end{equation}

Next, we prove the lower bound on $P(i \to j)$. By Lemma~\ref{Lemma: PEP}, 
\begin{equation}
P(i\to j) = \frac{1}{4\pi}\int_{-\infty}^{\infty}dw \frac{1}{w^2+1/4}\cdot \prod_{\overset{k=1}{\cos\theta^{(i,j)}_k<1}}^{K}\frac{1+\sigma_{\bs}^2}{\sigma_{\bs}^4\left(\omega^2 +\frac{1}{4}\right)\sin^2\theta^{(i,j)}_k + 1+\sigma_{\bs}^2}.
\end{equation}
Similarly, following the same argument as before with
\begin{equation}
\frac{\partial P(i\rightarrow j)}{\partial \sin^{2}\theta_k^{(i,j)}}<0 \quad \text{and}\quad \frac{\partial^2 P(i\rightarrow j)}{\partial \left(\sin^{2}\theta_k^{(i,j)}\right)^2}<0,
\end{equation}
$P(i\to j)$ is minimized if all the principal angles have the same value given $d^2_{i,j}=\sum_{k=1}^{K}\sin^2\theta^{(i,j)}_k$. This leads to:
\begin{align}
P(i\to j) &\geq \frac{1}{4\pi}\int_{-\infty}^{\infty}dw \frac{1}{w^2+1/4}\cdot \l(\frac{1}{\frac{4}{K}g(\sigma^2_{\bs})\left(\omega^2 +\frac{1}{4}\right) d^2_{i,j}+ 1}\r)^K,\\
& \geq  \frac{1}{4\pi}\int_{\omega^2 +\frac{1}{4}\leq 1}dw \frac{1}{w^2+1/4}  \cdot \l(\frac{1}{1+ \frac{4}{K}g(\sigma^2_{\bs})\left(\omega^2 +\frac{1}{4}\right) d^2_{i,j}}\r)^K, \nn\\
& \geq \frac{1}{4\pi}\int_{\omega^2 +\frac{1}{4}\leq 1}dw \frac{1}{w^2+1/4}  \cdot \l(\frac{1}{1+\frac{4}{K}g(\sigma^2_{\bs})d^2_{i,j}}\r)^K,\nn\\
& = \frac{1}{\pi} \arctan(\sqrt{3})\l(\frac{1}{1+\frac{4}{K}g(\sigma^2_{\bs})d^2_{i,j}}\r)^K.
\end{align}
This completes the proof.

\vspace{-3mm}
\subsection{Proof of Theorem~\ref{Theorem: Bounds_C2_Capacity}}\label{Proof: Bounds_C2_Capacity}
\vspace{-2mm}
First, we prove the lower bound on the $\epsilon$-classification capacity. By~\eqref{LB: C2_Capacity}, 
\begin{equation}
KL^{-\frac{2}{NK}}\log_2(1+ g(\sigma^2_{\bs}))^{-1} = \log\frac{2\epsilon}{L(1+ g(\sigma^2_{\bs}))}.
\end{equation}

For a high data SNR, it follows from the above equation that
\begin{equation}\label{Eq: C2_Upper}
L \gtrsim 2^{\frac{N}{2}\l(K\log_2K+ K\log_2\log_2(1+ g(\sigma^2_{\bs})) - K\log_2 \log_2\frac{1+g(\sigma^2_{\bs})}{2\epsilon}\r)}.
\end{equation}
Next, we prove the upper bound on the $\epsilon$-classification capacity. From~\eqref{UB: C2_Capacity},
\begin{equation}\label{Eq: UB_appendix_classifcation_capacity}
\frac{1}{3}\l(\frac{1}{1+\frac{4}{K}g(\sigma^2_{\bs})\delta^2_{\sf ub}}\r)^K = \epsilon.
\end{equation}
{As the direct approach is intractable, we find a lower bound on the right-hand side of~\eqref{Eq: UB_appendix_classifcation_capacity}. Using the fact that $ 4KL^{-\frac{2}{NK}} \geq \delta^2_{\sf ub}>1$ for large $N$ and $K$, 
\begin{equation}
\frac{1}{3}\l(\frac{1}{1+\frac{4}{K}g(\sigma^2_{\bs})\delta^2_{\sf ub}}\r)^K \geq\frac{1}{3}\l(\frac{1}{1+16g(\sigma^2_{\bs})L^{-\frac{2}{NK}}}\r)^K, \quad N,K\to \infty.
\end{equation}
Then,~\eqref{Eq: UB_appendix_classifcation_capacity} asymptotically reduces to
$\frac{1}{3}\l(\frac{1}{1+16g(\sigma^2_{\bs})L^{-\frac{2}{NK}}}\r)^K \approx \epsilon$, as $N,K \to \infty$.
This results in an asymptotic upper bound on the $\epsilon$-classification capacity:
\begin{equation}
L \lesssim 2^{\frac{N}{2}\l(K\log_24K + K\log_2\frac{4g(\sigma^2_{\bs})}{1-3\epsilon}\r)}.
\end{equation}
The substituting of $N = \beta R$ gives \eqref{Eq: bounds_classification_capacity}. Furthermore, as $R,K \to \infty$, the bounds on the $\epsilon$-classification capacity scale in~\eqref{Eq: scaling_classification_capacity_class-selection}, which completes the proof.

%-----------------------------------------------------------------------------------------------
\vspace{-3mm}
\subsection{Proof of Proposition~\ref{Proposition: bounds_Ergodic_capacity}}\label{Proof: Bounds_Ergodic_Capacity}
\vspace{-3mm}
1) Bounds on ergodic classification capacity: The lower bound in Theorem~\ref{Theorem: Bounds_C2_Capacity} can be rewritten as $  C^{*}(R) \gtrsim 2^{\frac{R}{B}\cdot \gamma_{\sf lb}}$, 
where $\gamma_{\sf lb} = \frac{\beta B}{2}\l(K\log_2{K}+ Kc_{\sigma_{\bs}} - Kc_{\epsilon}\r)$. It follows that
\begin{equation}\label{Eq: appendx_1}
\mathbb E[C^{*}(R)] \gtrsim  \mathbb E\l[\l(1+\rho|h|^2\r)^{\gamma_{\sf lb}}\r].
\end{equation} 
For a high SNR, 
\begin{equation}\label{Approx: high_SNR}
\mathbb E\l[\l(1+\rho|h|^2\r)^{\gamma_{\sf lb}}\r] \approx \rho^{\gamma_{\sf lb}} \mathbb E[|h|^{2\gamma_{\sf lb}}] \overset{(a)}{=} \Gamma(\gamma_{\sf lb}+1), \quad \rho \to \infty,
\end{equation}
where $(a)$ uses $|h|^2 = \exp(1)$ and $\Gamma(\cdot)$ denotes the gamma function. Given large $\gamma_{\sf lb}$ and using the \emph{stirling's apporximation}
\begin{equation}\label{Eq: appendx_2}
 \mathbb E[|h|^{2\gamma_{\sf lb}}]  \approx \sqrt{2\pi\gamma_{\sf lb}}\cdot  e^{\gamma_{\sf lb}(\log \gamma_{\sf lb}-1)},\quad  \gamma_{\sf lb} \gg 1.
\end{equation}
Combining the above result with~\eqref{Eq: appendx_1} and~\eqref{Approx: high_SNR}, ~\eqref{Eq: Lower_bound_ergodic_capacity} follows. Following the same procedure, the upper bound can be proved. 

2) Scaling law: Consider the ergodic communication rate $\bar{R}  = \mathbb E[B\log_2(1+\rho|h|^2)]$, $\bar{R} \to B\log_2\rho + B\mathbb E[\log_2|h|^2]$ and hence 
\begin{equation}
\lim_{\rho \to \infty}\frac{\bar{R}}{\log_2\rho} = B,\quad \rho \to \infty,
\end{equation}
implying $\bar{R} \to \infty$ as $\rho\to \infty$. Then, given sufficiently large $\rho$ and furthermore letting $K \to \infty$, both the derived bounds in~\eqref{Eq: Lower_bound_ergodic_capacity} scale as shown in~\eqref{Eq: scaling_ergodic_classification_capacity_class-selection} . This completes the proof.

\subsection{Proof of Proposition~\ref{Proposition: outage_classification_capacity}}\label{Proof: scaling_outage_Capacity}

The bounds in~\eqref{eqn: C2capacity_LB_fixed_model} are straightforward by substituting $R_{\delta}$ in ~\eqref{Eq: outage} into the outage classification capacity defined in~\eqref{Def: Outage_Classification_Capacity}. In the following, we prove the scaling law. To begin with, we show that as $\rho \to \infty$, $R_{\delta} \to \infty$. Given $R_{\delta} = \log_2\l(1 + \rho\log\l(\frac{1}{1-\delta}\r)\r)$, at a high SNR, one can have
\begin{equation}
R_{\delta} \approx  \log_2\l( \rho\log\l(\frac{1}{1-\delta}\r)\r) = \log_2\rho + \log_2\log\frac{1}{1-\delta}, \quad \rho \rightarrow \infty.
\end{equation}
This implies that $R_{\delta}$ scales linearly with $\log_2\rho$ given $\delta$. As a result, as $\rho \to \infty$, $R_{\delta} \to \infty$. Then, by letting $R_{\delta}, K\to \infty$, both bounds scale as shown in~\eqref{Eq: scaling_outage_classification_capacity_class-selection}, which completes the proof.

\subsection{Proof of Lemma~\ref{lemma:CDF_subspace_disntance_bounds}}\label{sec:appendix_proof_CDF_subspace_disntance_bounds}

To derive the upper bound on the CDF, namely $F_{d_{c}^2}(x)$, we fist obtain an upper bound on the probability of $\Pr(\sin^2\theta_{\max} < x)$. Given~\eqref{eqn:PDF_of_sin_square_theta_subspace},
\begin{equation}
\Pr(\theta_{\max} < x)=\frac{\Gamma\left(\frac{K+1}{2}\right)\Gamma\left(\frac{N-K+1}{2}\right)}{\sqrt{\pi}\Gamma\left(\frac{N+1}{2}\right)}(\sin x)^{K(N-K)} \ \tensor*[_2]{F}{_1}\l(\frac{N-K}{2}, \frac{1}{2};\frac{N + 1}{2};\sin^2x{\bf I}_{K}\r).
\end{equation}
Due to the fact that $\Pr(\sin^2\theta_{\max} < x) = \Pr(\theta_{\max} < \arcsin \sqrt{x})$, one can have 
\begin{align}\label{Ineq: relax_CDF_theta}
\Pr(\sin^2\theta_{\max} < x)  & = \frac{\Gamma\left(\frac{K+1}{2}\right)\Gamma\left(\frac{N-K+1}{2}\right)}{\sqrt{\pi}\Gamma\left(\frac{N+1}{2}\right)}x^{\frac{K(N-K)}{2}} \ \tensor*[_2]{F}{_1}\l(\frac{N-K}{2}, \frac{1}{2};\frac{N + 1}{2};x{\bf I}_{K}\r),\nn\\
& \leq  \frac{\Gamma\left(\frac{K+1}{2}\right)\Gamma\left(\frac{N-K+1}{2}\right)}{\sqrt{\pi}\Gamma\left(\frac{N+1}{2}\right)}x^{\frac{K(N-K)}{2}} \ \tensor*[_2]{F}{_1}\l(\frac{N-K}{2}, \frac{1}{2};\frac{N + 1}{2};{\bf I}_{K}\r),
\end{align} 
where the inequality uses the fact that $\tensor*[_2]{F}{_1}$ is a non-decreasing function in $x$. On the other hand, according to~\cite{Richards2020Arxiv}, one can have

\begin{equation}\label{eq: hyper_to_gamma}
\ \tensor*[_2]{F}{_1}\l(\frac{N-K}{2}, \frac{1}{2};\frac{N + 1}{2};{\bf I}_{K}\r) = \frac{\Gamma\left(\frac{N+1}{2}\right)\Gamma\left(\frac{1}{2}\right)}{\Gamma\l(\frac{K+1}{2}\r)\Gamma\left(\frac{N-K+1}{2}\right)}.
\end{equation}
Then, by substituting~\eqref{eq: hyper_to_gamma} into~\eqref{Ineq: relax_CDF_theta}, 
\begin{equation}\label{Upper:CDF_sin_theta}
\Pr(\sin^2\theta_{\max} < x)  = \Pr(\theta_{\max} < \arcsin x) \leq x^{\frac{K(N-K)}{2}}.
\end{equation}  

As two random subspaces of dimension $K$ embedded in $\mathbb R^{N}$ are quasi-orthogonal, given large $N$, the squared chordal distance $d^2_c$ can be approximated as $K\sin^2\theta_{\max}$. Then, we can bound the said CDF, namely $F_{d_{c}^2}(x)$, as 
\begin{align}\label{Eq: appendix_E}
F_{d_{c}^2}(x) = \Pr(d^2_c<x) \approx \Pr(K\sin^2\theta_{\max}<x) = \Pr\l(\sin^2\theta_{\max}<\frac{x}{K}\r), \quad N \to \infty.
\end{align}
By combining~\eqref{Upper:CDF_sin_theta} and ~\eqref{Eq: appendix_E}, the desired result follows.

\vspace{-3mm}
\subsection{Proof of Lemma~\ref{lemma:Expectation_PEP_bounds}}
\vspace{-3mm}
\label{sec:appendix_proof_Expectation_PEP_bounds}
It follows from~\eqref{Upper: PEP} that 

\begin{eqnarray}\label{eqn:app_mean_PEP_UB_unsimplified}
\mathbb E\left[P(i\to j)\right] &\leq& \mathbb E_{d_{c}^2}\left[\frac{1}{2}\left(\frac{1}{1+g(\sigma_{\bs}^2)}\right)^{\lfloor d^2_{i,j}\rfloor}\right] = \int_{0}^{K}\frac{1}{2}\left(\frac{1}{1+g(\sigma_{\bs}^2)}\right)^{\lfloor x \rfloor} \text{PDF}_{d^2_{c}}(x)dx,\nonumber \\
&\leq&\frac{1}{2}\left(\frac{1}{1+g(\sigma_{\bs}^2)}\right)^{K} + \frac{\log(1+g(\sigma_{\bs}^2))}{2}\int_{0}^{K}F_{d^2_{c}}^{\sf{ub}}(x)\left(\frac{1}{1+g(\sigma_{\bs}^2)}\right)^{x-1}dx.
\end{eqnarray}
Combining~\eqref{eqn:CDF_subspace_distance_UB} and~\eqref{eqn:app_mean_PEP_UB_unsimplified}, gives
\begin{align}
\mathbb E\left[P(i\to j)\right] \leq \frac{1}{2}\left(\frac{1}{1+g(\sigma_{\bs}^2)}\right)^{K} +\frac{K\log(1+g(\sigma_{\bs}^2))}{2}\int_{0}^{1}x^{\frac{K(N-K)}{2}}\frac{1}{\left[1+g(\sigma_{\bs}^2)\right]^{Kx-1}}dx.\label{Ineq:pep_upper_temp}
\end{align}
We decompose the second term at the RHS of~\eqref{Ineq:pep_upper_temp} as follows
\begin{align}
& \frac{K\log(1+g(\sigma_{\bs}^2))}{2}\l[\int_{0}^{\frac{2}{K}}x^{\frac{K(N-K)}{2}}\left(\frac{1}{1+g(\sigma_{\bs}^2)}\right)^{Kx-1}dx + \int_{\frac{2}{K}}^{1}x^{\frac{K(N-K)}{2}}\left(\frac{1}{1+g(\sigma_{\bs}^2)}\right)^{Kx-1}dx\r]\nn\\
& \leq \frac{K\log(1+g(\sigma_{\bs}^2))}{2}\l[ \left(1+g(\sigma_{\bs}^2)\right)\int_{0}^{\frac{2}{K}}x^{\frac{K(N-K)}{2}}dx + \left(\frac{1}{1+g(\sigma_{\bs}^2)}\right)\int_{\frac{2}{K}}^{1}x^{\frac{K(N-K)}{2}}dx\r]\label{Eq:second_term}.
\end{align}
For large $N$, \eqref{Eq:second_term} can be asymptotically expressed as $\frac{\log(1+g(\sigma_{\bs}^2))}{1+g(\sigma_{\bs}^2)}\frac{1}{N}$. Substituting it into~\eqref{Ineq:pep_upper_temp}, \eqref{inqe:avg_PEP_UB} follows. This completes the proof.

\newpage
%------------------------------------------

\end{document}

%% file: header.tex
\newtheorem{theorem}{Theorem}

\newtheorem{definition}{Definition}

\newtheorem{lemma}{Lemma}

\newtheorem{proposition}{Proposition}

\newtheorem{assumption}{Assumption}

\newtheorem{remark}{\bf Remark}

\def\qed{$\Box$}

\def\proof{\noindent{\emph{Proof:} }}

\def
\endproof{\hspace*{\fill}~\qed
\par
\endtrivlist\unskip}
\def\endproof{\hspace*{\fill}~\qed\par\endtrivlist\vskip3pt}

\def\phi{\varphi}

\def\l{\left}
\def\r{\right}
\def\({\left(}
\def\){\right)}

\def\bs{{\mathbf{s}}}

\def\bw{{\mathbf{w}}}
\def\bx{{\mathbf{x}}}

\def\b0{{\mathbf{0}}}

\def\bF{{\mathbf{F}}}

\def\bI{{\mathbf{I}}}

\def\bU{{\mathbf{U}}}

\newcommand{\nn}{\nonumber}

%% file: Remote_Classification_Arxiv_Submit.bbl
\begin{thebibliography}{00}

\bibitem{Niyato2020CommSurvey} X. Wang, Y. Han, V. C. M. Leung, D. Niyato, X. Yan, and X. Chen, ``Convergence of edge computing and deep learning: A comprehensive survey," \emph{IEEE Commun. Surveys Tuts.}, vol. 22, no. 2, pp. 869-904, 2020.
\bibitem{Xu2019VTM} M. Boban, A. Kousaridas, K. Manolakis, J. Eichinger, and W. Xu, ``Connected roads of the future: Use cases, requirements, and design considerations for vehicle-to-everything communications," \emph{IEEE Veh. Tech. Magazine}, vol. 13, no. 3, pp. 110-123, Sep. 2018

\bibitem{Berger1971Book} T. Berger, \emph{Rate distortion theory: A mathematical basis for data compression}. Englewood Cliffs, N.J.: Prentice-Hall, 1971.

\bibitem{Cover1991Book} T. M. Cover and J. A. Thomas, \emph{Elements of information theory}. New York: Wiley, 1991.

\bibitem{Bishop2006Book} C. M. Bishop, \emph{Pattern recognition and machine learning}. New York: Springer, 2006.


\bibitem{Nokleby2015TIT} M. Nokleby, M. Rodrigues, and R. Calderbank, ``Discrimination on the Grassmann manifold: Fundamental limits of subspace classifiers,'' \emph{IEEE Trans. Inf. Theory}, vol. 61, no. 4, pp. 2133-2147, Apr. 2015.



\bibitem{You2017CommSurvey} Y. Mao, C. You, J. Zhang, K. Huang, and K. B. Letaief, ``A survey on mobile edge computing: The communication perspective," \emph{IEEE Commun. Surveys Tuts.}, vol. 19, no. 4, pp. 2322-2358, 2017.

\bibitem{You2017TWC} C. You, K. Huang, H. Chae, and B. Kim, ``Energy-efficient resource allocation for mobile-edge computation offloading," \emph{IEEE Trans. Wireless Commun.}, vol. 16, no. 3, pp. 1397-1411, Mar. 2017.

\bibitem{Quek2017TCOM} T. Q. Dinh, J. Tang, Q. D. La, and T. Q. S. Quek, ``Offloading in mobile edge computing: Task allocation and computational frequency scaling," \emph{IEEE Trans. Commun.}, vol. 65, no. 8, pp. 3571-3584, Aug. 2017.

\bibitem{Niyato2012TWC} D. Huang, P. Wang, and D. Niyato, ``A dynamic offloading algorithm for mobile computing," \emph{IEEE Trans. Wireless Commun.}, vol. 11, no. 6, pp. 1991-1995, Jun. 2012.

\bibitem{Zhang2017TWC} Y. Mao, J. Zhang, S. H. Song, and K. B. Letaief, ``Stochastic joint radio and computational resource management for multi-user mobile-edge computing systems," \emph{IEEE Trans. Wireless Commun.}, vol. 16, no. 9, pp. 5994-6009, Sep. 2017.

\bibitem{Tao2020TWC} K. Li, M. Tao, and Z. Chen, ``Exploiting computation replication for mobile edge computing: A fundamental computation-communication tradeoff study," to appear in \emph{IEEE Trans. Wireless Commun.}.

\bibitem{Niyato2020TMC} A. {Ndikumana}, N. H. {Tran}, T. M. {Ho}, Z. {Han}, W. {Saad}, D. {Niyato} and C. S. {Hong}, ``Joint communication, computation, caching, and control in big data multi-access edge computing," \emph{IEEE Trans. Mobile Comput.}, vol. 19, no. 6, pp. 1359-1374, 1 Jun. 2020.

\bibitem{Zhang2016ISIT} J. Liu, Y. Mao, J. Zhang, and K. B. Letaief, ``Delay-optimal computation task scheduling for mobile-edge computing systems," in \emph{Proc. IEEE Int. Symp. Inf. Theory}, Barcelona, Spain, July 10-15, 2016.

\bibitem{Shi2020JCIN} X. Yang, S. Hua, Y. Shi, H. Wang, J. Zhang, and K. B. Letaief, ``Sparse optimization for green edge AI inference," \emph{J. Commun. Info. Netw.},vol. 5, no. 1, pp. 1-15, Mar. 2020.

\bibitem{Zhou2019Arxiv} W. Shi, Y. Hou, S. Zhou, Z. Niu, Y. Zhang, and L. Geng, ``Improving device-edge cooperative inference of deep learning via 2-step pruning," \emph{[online]. Available: https://arxiv.org/pdf/1903.03472.pdf}, 2019.

\bibitem{Zhang2020Arxiv} J. Shao and J. Zhang, ``Communication-computation trade-off in resource-constrained edge inference," \emph{[online]. Available: https://arxiv.org/pdf/2006.02166.pdf}, 2020.

\bibitem{Chen2018MECOMM} E. Li, Z. Zhou, and X. Che, ``Edge intelligence: On-demand deep learning model co-inference with device-edge synergy," in \emph{Proc. ACM Workshop Mobile Edge Commun. (MECOMM'18)}, Budapest, Hungary, Aug. 20-25, 2018.

\bibitem{Katti2018HotNets} S. P. Chinchali, E. Cidon, E. Pergament, T. Chu, and S. Katti, ``Neural networks meet physical networks: Distributed inference between edge devices and the cloud," in \emph{Proc. ACM Workshop Hot Topics Netw. (HotNets'18)}, Redmond, Washington, Nov. 15-16, 2018.
\bibitem{Yu2017SIGIR} W. Yu, Z. Sun, H. Liu, Z. Li, and Z. Zheng, ``Multi-level deep learning based e-Commerce product categorization," in \emph{Proc. ACM SIGIR 2018 Workshop on eCommerce,}  Ann Arbor, Michigan, July 8-12, 2018.
\bibitem{Malerba2007JIIS} M. Ceci and D. Malerba, ``Classifying web documents in a hierarchy
of categories: A comprehensive study'', \emph{J. Intell. Info. Sys.}, vol 28, no. 1, pp. 37-78, Jan. 2017.


\bibitem{Conway1987Book} J. H. Conway, N. J. A. Sloane, and E. Bannai, \emph{Sphere packings, lattices, and groups}. New York: Springer, 1987.
\bibitem{Nogin2002TIT} A. Barg and D. Y. Nogin, ``Bounds on packings of spheres in the Grassmann manifold,'' \emph{IEEE Trans. Inf. Theory}, vol. 48, no. 9, pp. 2450-2454, Sep. 2002.

\bibitem{Marzetta2000TIT} B. M. Hochwald and T. L. Marzetta, ``Unitary space-time modulation for multiple-antenna communications in Rayleigh flat fading," \emph{IEEE Trans. Inf. Theory}, vol. 46, no. 2, pp. 543-564, Mar. 2000


\bibitem{b5} P.-A. Absil, A. Edelman, and P. Koev, ``On the largest principal angle between random subspaces,'' \emph{Linear Algebra and Its Applications}, vol. 141, no. 1, pp. 288-294, Apr. 2006.

\bibitem{Bennis2020Arxiv} A. Elgabli, J. Park, C. B. Issaid, and M. Bennis, ``Harnessing wireless channels for scalable and privacy-preserving federated learning,'' \emph{[online]. Available: https://arxiv.org/pdf/2007.01790.pdf}.

\bibitem{Niyato2019Arxiv} H. T. Nguyen, N. C. Luong, J. Zhao, C. Yuen, and D. Niyato, ``Resource allocation in mobility-aware federated learning networks: A deep reinforcement learning approach,'' \emph{[online]. Available: https://arxiv.org/pdf/1910.09172.pdf}.

\bibitem{ConwayPack} J. H. Conway, R. H. Hardin, and N. J. A. Sloane, ``Packing lines, planes, etc.: Packings in Grassmannian spaces," \emph{Experimental Mathematics}, vol. 5, no. 2, pp. 139-159, Feb. 1996.

\bibitem{Richards2020Arxiv} D. Richards, and Q. Zhang, ``A reflection formula for the Gaussian hypergeometric function of matrix argument,'' \emph{[online]. Available: https://arxiv.org/pdf/2002.05248.pdf}.

\end{thebibliography}
